\documentclass{lmcs}
\pdfoutput=1

\usepackage{lastpage}
\lmcsdoi{17}{1}{7}
\lmcsheading{}{\pageref{LastPage}}{}{}%
{Dec.~17,~2018}{Jan.~29,~2021}{}

\pdfoutput=1
\usepackage[utf8]{inputenc}

\usepackage[all]{xy}
\usepackage{mathpartir}
\usepackage{amsmath}
\usepackage{ulem}
\usepackage{graphics}
\usepackage{amssymb}

\usepackage{prooftree}
\usepackage{float}
\usepackage{thm-restate}

\theoremstyle{plain}
\newtheorem{theorem}[thm]{Theorem}

\newtheorem{lemma}[thm]{Lemma}

\newtheorem{corollary}[thm]{Corollary}

\theoremstyle{definition}

\newtheorem{definition}[thm]{Definition}

\newtheorem{property}[thm]{Property}
\newtheorem{example}[thm]{Example}
\newtheorem{notation}[thm]{Notation}
\usepackage{framed}
\usepackage{microtype}
\long\def\ignore#1{\relax}
\newcommand{\s}{{\tt t}}

\newcommand{\uu}{{\tt u}}
\newcommand{\vv}{{\tt v}}
\newcommand{\x}{{\tt x}}
\newcommand{\y}{{\tt y}}
\newcommand{\z}{{\tt z}}
\newcommand{\w}{{\tt w}}
\newcommand{\ap}{{\tt a}}
\newcommand{\bp}{{\tt b}}
\newcommand{\cp}{{\tt c}}
\newcommand{\Ap}{{\mathcal A}}
\newcommand{\Lp}{{\Lambda_\p}}

\newcommand{\p}{{\tt p}}
\newcommand{\q}{{\tt q}}

\newcommand{\dup}{\delta}

\newcommand{\Rew}[1]{\rightarrow_{#1}}

\newcommand{\Rewn}[1]{\rightarrow^*_{#1}}

\newcommand{\reds}{\rightarrow^*}

\newcommand{\pair}[2]{\langle #1,#2 \rangle }

\newcommand{\A}{{\tt A}}
\newcommand{\D}{{\tt B}}
\newcommand{\E}{{\tt E}}
\newcommand{\arrow}{\rightarrow}

\newcommand{\der}{\vdash}

\newcommand{\sep}{\hspace{.5cm}}
\newcommand{\minisep}{\hspace{.4cm}}

\newcommand{\Pu}{{\mathcal P}}
\newcommand{\dem}{\triangleright}

\newcommand{\multiset}[1]{[#1]}

\newcommand{\emul}{\mult{\, }}
\newcommand{\set}[1]{\{ #1 \}}
\newcommand{\ie}{{\it i.e.}}
\newcommand{\wlg}{{without loss of generality}}

\newcommand{\eg}{{\it e.g.}}
\newcommand{\cf}{{\it cf.}}
\newcommand{\ih}{{\it i.h.}}
\newcommand{\isubs}[2]{\{ #1 / #2 \}}

\newcommand{\Gam}{\Gamma}

\newcommand{\Lam}{\Lambda}
\newcommand{\Del}{\Delta}
\newcommand{\del}{\del}

\newcommand{\fv}[1]{{\tt fv}(#1)}
\newcommand{\bv}[1]{{\tt bv}(#1)}

\newcommand{\dom}[1]{{\tt dom}(#1)}

\newcommand{\oprod}{o}

\newcommand{\pathd}[1]{\times_{2}}

\newcommand{\pder}{\Vdash}
\newcommand{\sm}{\setminus}

\newcommand{\es}{\emptyset}

\newcommand{\slist}{{\tt L}}

\newcommand{\id}{{\tt I}}
\newcommand{\fail}{{\tt fail}}
\newcommand{\deft}[1]{{\bf #1}}

\newcommand{\run}{{\tt dB}} 
\newcommand{\rdeux}{{\tt subs}} 
\newcommand{\rtrois}{{\tt match}_s} 
\newcommand{\rquatre}{{\tt match}_f} 
\newcommand{\rcinq}{{\tt rem}_f} 
\newcommand{\rsix}{{\tt la}_f} 
\newcommand{\rsept}{{\tt abs}_f} 
\newcommand{\rhuit}{{\tt app}_f} 
\newcommand{\rneuf}{{\tt lem}_f} 

\newcommand{\I}{{\cal I}}

\newcommand{\Reweq}[1]{\stackrel{=}{\Rew{#1}}}
\newcommand{\sig}{\sigma}

\newcommand{\canonical}{canonical form}
\newcommand{\canonicals}{canonical forms}

\newcommand{\can}{{\tt cf}}
\newcommand{\anf}{\mbox{anf}}
\newcommand{\Anf}{\mbox{Anf}}

\newcommand{\ccontext}{{\tt C}}
\newcommand{\hcontext}{{\tt H}}

\newcommand{\ccount}[1]{\#(#1)} 
\newcommand{\iI}{{i \in I}}

\newcommand{\kK}{k  \in K}

\newcommand{\bu}{\bigvee}
\newcommand{\C}{{\tt C}}

\newcommand{\fin}[1]{{\tt F}(#1)}

\newcommand{\donne}{\triangleright}

\newcommand{\Final}{{\tt Final}}
\newcommand{\Head}{{\tt Head}}

\newcommand{\Prefix}{{\tt Prefix}}

\newcommand{\Varp}{{\tt Varp}}
\newcommand{\Pairp}{{\tt Pairp}}
\newcommand{\Abs}{{\tt Abs}}

\newcommand{\Unionk}{{\tt Many}}

\newcommand{\Prod}{{\tt Prod}}

\newcommand{\Subs}{{\tt Subs}}

\newcommand{\Pa}[2]{{\tt P}^{#1}_{#2}}
\newcommand{\K}{{\tt T}}

\newcommand{\KI}{{\tt M}}
\newcommand{\LK}[2]{{\tt H}_{#1}^{#2}}
\newcommand{\V}{{\cal V}}
\newcommand{\meas}{{\tt meas}}

\newcommand{\altp}{{\cal A}_0}

\newcommand{\borne}{\triangledown}
\newcommand{\R}{{\cal R}}

\newcommand{\toc}[1]{{\tt toc}(#1)}

\newcommand{\munion}{\sqcup}

\newcommand{\trvarpat}{\mathtt{varpat}}

\newcommand{\trpairpat}{\mathtt{pairpat}}

\newcommand{\trsub}{\mathtt{sub}}

\newcommand{\quadd}{\hspace{.2cm}}
\newcommand{\quadt}{\hspace{.4cm}}

\newcommand{\appctx}[2]{#1 {[ }  #2 { ]} }
\newcommand{\PT}{ \pi }
\newcommand{\prodt}[2]{\times(#1,#2)}

\newcommand{\rew}{\rightarrow}
\newcommand{\mult}[1]{[#1]}
\newcommand{\seq}[2]{#1 \vdash #2}
\newcommand{\ax}{\mathtt{ax}}
\newcommand{\Gamk}{\Gam_k}  
\newcommand{\sigk}{\sigma_k} 
\newcommand{\many}{\mathtt{many}}

\newcommand{\introarrow}{{\tt abs}}
\newcommand{\app}{{\tt app}}
\newcommand{\trpair}{\mathtt{pair}}
\newcommand{\tree}{{\cal T}}

\newcommand{\cal}[1]{\mathcal #1}

\begin{document}

\title{Solvability = Typability + Inhabitation}

\author{Antonio Bucciarelli\rsuper{a}}
\address{\lsuper{a}Universit\'e de  Paris, CNRS, IRIF, France  }	
\email{buccia@irif.fr} 

\author{ Delia Kesner\rsuper{b} }	
\address{\lsuper{b}Universit\'e de  Paris, CNRS, IRIF and Institut Universitaire de France, France  }
\email{kesner@irif.fr}  

\author{Simona Ronchi Della Rocca\rsuper{c}}	
\address{\lsuper{c}Dipartimento di Informatica, Universit\`a di Torino, Italy}	
\email{ronchi@di.unito.it}  


\begin{abstract}
We extend the classical notion of solvability to a $\lambda$-calculus equipped with pattern matching.
  We prove
  that solvability can be characterized by means of {\it typability}
  and {\it inhabitation} in an {\it intersection type system} $\Pu$
  based on {\it non-idempotent} types.  We show first that the system
  $\Pu$ characterizes the set of terms having \canonical, \ie\ that a
  term is typable if and only if it reduces to a \canonical.  But the
  set of solvable terms is properly contained in the set of
  \canonicals. Thus, typability alone is not sufficient to
  characterize solvability, in contrast to the case for the
  $\lambda$-calculus.  We then prove that typability, together with
  inhabitation, provides a full characterization of solvability, in
  the sense that a term is solvable if and only if it is typable and
  the types of all its arguments are inhabited.  We complete the
  picture by providing an algorithm for the inhabitation problem of
  $\Pu$.
  \end{abstract}

\maketitle

\section{Introduction}

In these last years there has been a growing interest in {\it pattern
  $\lambda$-calculi}~\cite{PeytonJones,Kahl-2004a,CK04,KvOdV08,JK09,Petit09}
which are used to model the pattern-matching primitives of functional
programming languages (\eg\ OCAML, ML, Haskell) and proof assistants
(\eg\ Coq, Isabelle).  These calculi are extensions of the $\lambda$-calculus:
abstractions are written as $\lambda \p. \s$, where $\p$ is a {\it
  pattern} specifying the expected structure of the argument.
In this paper we restrict our
attention to {\it pair} patterns, which are expressive enough to
illustrate the challenging notion of solvability in the
framework of pattern $\lambda$-calculi.

We define a calculus with {\it
  explicit pattern-matching} called $\Lp$.  The use of
explicit pattern-matching becomes very appropriate to implement
different {\it evaluation strategies}, thus giving rise to different
{\it programming languages} with
pattern-matching~\cite{CK04,CirsteaFK07,Bal2012}.  In all of them, an
application $(\lambda \p. \s)\uu$ reduces to $\s[\p/\uu]$, where the
constructor $[\p/\uu]$ is an explicit matching, defined by means of
suitable reduction rules, which are used to decide if the argument
$\uu$ matches the pattern $\p$.  If the matching is possible, the
evaluation proceeds by computing a substitution which is applied to
the body $\s$. Otherwise, two cases may arise: either a successful
matching is not possible at all, and then the term $\s[\p/\uu]$
reduces to a {\it failure}, denoted by the constant $\fail$, or
pattern matching could potentially become possible after the
application of some pertinent substitution to the argument $\uu$, in
which case the reduction is simply {\it blocked}.  For example,
reducing $(\lambda \pair{\z_1}{\z_2}.\z_1)(\lambda \y.\y)$ leads to a failure,
while reducing $(\lambda \pair{\z_1}{\z_2}.\z_1)\y$ leads to a blocking
situation.

We aim to study {\it solvability} in the
$\Lp$-calculus. Let us first recall this  notion
in the framework of the $\lambda$-calculus: a closed
(\ie, without free variables) $\lambda$-term $\s$ is solvable if there
is $n\geq 0$ and there are terms $\uu_{1},...,\uu_{n}$ such that
$\s\uu_{1}...\uu_{n}$ reduces to the identity function.  Closed
solvable terms represent meaningful programs: if $\s$ is closed and
solvable, then $\s$ can produce any desired result when applied to a
suitable sequence of arguments. The relation between solvability and
meaningfulness is also evident in the semantics: it is sound to equate
all unsolvable terms, as in Scott's original model
$D_\infty$~\cite{Scott70}. This notion can be easily extended to open
terms, through the notion of {\it head context}, which does the job of
both closing the term and then applying it to an appropriate sequence
of arguments. Thus a $\lambda$-term $\s$ is solvable if there is a
head context $\hcontext$ such that, when $\hcontext$ is filled by
$\s$, then $\appctx{\hcontext}{\s}$ is closed and reduces to the
identity function.

In order to extend the notion of solvability to the 
$\Lp$-calculus, it is clear that pairs have to be taken into account.
A relevant question is whether a pair should be considered as
meaningful.  At least two choices are possible: a {\it lazy} semantics
considering any pair to be meaningful, or a {\it strict} one requiring
both of its components to be meaningful. We chose a lazy approach, in
fact in the operational semantics of $\Lp$ the constant $\fail$ is
different from $\pair{\fail}{\fail}$: if a term reduces to $\fail$ we
do not have any information about its result, but if it reduces to
$\pair{\fail}{\fail}$ we know at least that it represents a pair. In
fact, being a pair is already an observable property, which in
particular is sufficient to unblock an explicit matching,
independently from the solvability  of
its components.  As a consequence, a term $\s$ is {\it defined} to be
{\it solvable} iff there exists a head context
$\hcontext$ such that $\appctx{\hcontext}{\s}$ is closed and reduces
to a pair. Thus for example, the term $\pair{\s}{\s}$ is always
solvable, also when $\s$ is not
solvable. Our notion of solvability turns out to be conservative with respect
to the  same notion for the $\lambda$-calculus (see
Theorem~\ref{thm:con}).  \medskip

\medskip

In this paper we characterize solvability for
the $\Lp$-calculus through two different and complementary notions
related to a type assignment system with non-idempotent intersection
types, called $\Pu$. The first one is {\it typability}, that gives the
possibility to construct a typing derivation for a given term, and the
second one is {\it inhabitation}, which gives the possibility to
construct a term from a given typing.  More precisely, we first supply
a notion of {\it canonical form} such that reducing a term to some
canonical form is a {\it necessary} but {\it not a sufficient}
condition for being solvable. In fact, canonical
forms may contain blocking explicit matchings, so that we need to
guess whether or not there exists a substitution being able to {\it
  simultaneously} unblock all these blocked forms. Our type system
$\Pu$ {\it characterizes} canonical forms: a term $\s$ has a canonical
form if and only if it is typable in system $\Pu$
(Theorem~\ref{l:characterization-canonical}). Types are of the shape
$\A_1 \to \A_2 \to...\to \A_n \to \sigma$, for $n \geq 0$, where
$\A_{i}$ are multisets of types and $\sigma$ is a type. The use of
multisets to represent the non-idempotent intersection is standard,
namely $\mult{\sig_{1},...,\sig_{m}}$ is just a notation for
$\sig_{1}\cap...\cap\sig_{m}$. By using the type system $\Pu$ we can
supply the following {\it characterization} of
solvability (Theorem~\ref{t:main-result}): a closed term $\s$ in the
$\Lp$-calculus is solvable if and only if $\s$ is
typable in system $\Pu$, let say with a type of the shape $\A_1 \to
\A_2 \to...\to \A_n \to \sigma$ (where $\sigma$ is a type derivable for a pair),
and for all $1\leq i \leq n$ there is a term $\s_i$ inhabiting the
type $\A_i$. In fact, if $\uu_{i}$ inhabits the type $\A_{i}$, then
$\s \uu_{1}...\uu_{n}$, resulting from plugging $\s$ into the head
context $\Box \uu_{1}...\uu_{n}$, reduces to a pair.  The extension of
this notion to open terms is obtained by suitably adapting the notion
of head context of the $\lambda$-calculus to our pattern calculus.
         
The property of being solvable in our calculus is clearly
undecidable. More precisely, the property of having a canonical
form is undecidable, since $\Lp$ extends the $\lambda$-calculus, the
$\lambda$-terms having a $\Lp$-canonical form are exactly the solvable ones,
and solvability of $\lambda$-terms is an undecidable property.  But
our characterization of solvability through the inhabitation property
of $\Pu$ does not add a further level of undecidability: in fact we
prove that inhabitation for system $\Pu$ is {\it decidable}, by
designing a sound and complete inhabitation algorithm for it.  The
inhabitation algorithm presented here is a non trivial extension of
the one given in~\cite{bkdlr14, DBLP:journals/lmcs/BucciarelliKR18}
for the $\lambda$-calculus, the difficulty of the extension being due to
the explicit pattern matching.

\paragraph{\bf Relation with $\lambda$-calculus}

Let us
recall the existing characterizations of solvability for the $\lambda$-calculus:
\begin{description}
 \item[(1)] $\appctx{\hcontext}{\s}$ reduces to the identity for an appropriate head context $\hcontext$;  
\item[(2)] $\s$ has a head normal form; 
\item[(3)]  $\s$ can be typed in a
suitable intersection type system. 
\end{description}
Statement {\bf (1)} is the definition of solvability, Statement {\bf
  (2)} (resp. {\bf (3)}) is known as the {\it syntactical} (resp. {\it
  logical}) characterization of solvability.  The syntactical
characterization, \ie\ {\bf (2)} $\Leftrightarrow$ {\bf (1)} has been
proved in an untyped setting using the standardization theorem
(see~\cite{barendregt84nh}).  The logical characterization, \ie\ {\bf
  (3)} $\Leftrightarrow$ {\bf (1)}, uses the syntactical one: it is
performed by building an intersection type assignment system
characterizing terms having head normal form (see for
example~\cite{Dezani-Ghilezan:TYPES-2002}). Then the implication {\bf
  (3)} $\Rightarrow$ {\bf (2)} corresponds to the soundness of the
type system (proved by means of a subject reduction property), while
{\bf (2)} $\Rightarrow$ {\bf (3)} states its completeness (proved by
subject expansion).

Traditional systems in the literature characterizing solvability for
$\lambda$-calculus are for example~\cite{tipoA-BCD:JSL,krivine93book}, where
intersection is {\it idempotent}. Exactly the same
results hold for {\it non-idempotent} intersection types, for example
for the type system~\cite{deC09,bkdlr14}, which is a
restriction of $\Pu$ to $\lambda$-terms.

How does the particular case of the $\lambda$-calculus fit in the
``solvability = typability + inhabitation'' frame?  We address this
issue in the following digression.  Let $\iota$ be a type which is
peculiar to some subset of ``solvable'' terms, in the sense that any
closed term of type $\iota$ reduces to a solvable term in that set (in
the present work, such a subset contains all the pairs, in the case of
the $\lambda$-calculus, it is
  the singleton containing only the identity). Then a type $\tau$ of
the form $\A_1\arrow\ldots\arrow\A_n\arrow\iota$ may be viewed as a
certificate, establishing that, by applying a closed term $\s:\tau$ to
a sequence of closed arguments $\uu_i:\A_i$, one gets a term that
reduces to a term in such a subset.  This is summarized by the slogan
``solvability = typability + inhabitation''.
In the case of the call-by-name $\lambda$-calculus, however,
typability alone already guarantees
solvability. The mismatch is only apparent,
though: any closed, head normal term of the $\lambda$-calculus,
\ie\ any term of the shape $\lambda \x_1\ldots\x_n.\x_j\s_1\ldots\s_m\ (n,m
\geq 0)$, may be assigned a type of the form
$\A_1\arrow\ldots\arrow\A_n\arrow\iota$
where all the $\A_i$'s are empty except the one corresponding to the
head variable $\x_j$, which is of the shape
$[\underbrace{\emul\arrow\ldots\arrow\emul}_{m}\arrow\iota]$.  The problems of finding inhabitants of the empty type
and of $\underbrace{\emul\arrow\ldots\arrow\emul}_{m}\arrow\iota$ are
both trivial.  Hence, ``solvability = typability + inhabitation''
does hold for the $\lambda$-calculus, too, but the ``inhabitation'' part
is trivial in that particular case. This is due, of course, to the
fact that the head normalizable terms of the $\lambda$-calculus coincide
with both the solvable terms and the typable ones.

 But in other settings, a term may be both typable and non
 solvable, the types of (some of) its arguments
 being non-inhabited (Theorem~\ref{th:completess}).

\paragraph{\bf Related work} 

This work is an expanded and revised version of~\cite{DBLP:conf/tlca/BucciarelliKR15}. In particular: 
\begin{itemize}
\item The reduction relation on $\Lp$-terms in this paper is
  smaller. In particular, the new reduction system uses {\it reduction at a
    distance}~\cite{AK10}, implemented through the notion of list
  contexts.
  \item Accordingly, the  type system $\Pu$ in this paper
    and the corresponding inhabitation algorithm are much simpler.
 In particular, the use of idempotent/persistent  information
    on  the structure of patterns is no more needed. 
\end{itemize}
Non-idempotent intersection types are also used in~\cite{BernadetTh}
to derive strong normalization of a call-by-name calculus with
constructors, pattern matching and fixpoints. A similar result can be
found in~\cite{BBM18}, where the completeness proof of the (strong)
call-by-need strategy in~\cite{BBBK17} is extended to the case of
constructors. Based on~\cite{DBLP:conf/tlca/BucciarelliKR15}, the type
assignment system $\Pu$ was developed in~\cite{Alves} in order to
supply a quantitative analysis (upper bounds and exact measures) for
head reduction.

\medskip

\paragraph{\bf Organization of the paper.}
Section~\ref{s:calculus} introduces the pattern calculus and its main
properties. Section~\ref{s:type-system} presents the type system and
proves a characterization of terms having canonical forms by means of
typability. Section~\ref{s:inhabitation} presents a sound and complete
algorithm for the inhabitation problem associated with our typing
system. Section~\ref{s:charact} shows a complete characterization of
solvability using the inhabitation result and
the typability notion.  Section~\ref{s:conclusion} concludes by
discussing some future work.
\section{The Pair Pattern Calculus}
\label{s:calculus}

We now introduce the $\Lp$-calculus,
a generalization of the $\lambda$-calculus where
abstraction is extended to {\it patterns} and terms 
to {\it pairs}. Pattern matching is specified by means 
of  an  {\it explicit} 
operation.
Reduction is performed only if  the
argument matches the abstracted pattern. 

\medskip\noindent
\deft{Terms}  and \deft{contexts} of the $\Lp$-calculus are defined
by means of the following grammars: 
$$
\begin{array}{llll}
   \mbox{(\deft{Patterns})} &  \p,\q   & ::= & \x \mid  \pair{\p}{\q}  \\
  \mbox{(\deft{Terms})}    &  \s,\uu,\vv & ::= & \x \mid \lambda \p.\s \mid
                                                 \pair{\s}{\uu} \mid \s\uu \mid \s[\p/\uu] \mid \fail\\
 \mbox{(\deft{List Contexts})}     &  \slist & ::= & \Box  \mid \slist [\p/\s] \\
 \mbox{(\deft{Term Contexts})}    &  \ccontext & ::= & \Box   \mid \lambda \p. \ccontext \mid \pair{\ccontext}{\s} \mid \pair{\s}{\ccontext} \mid  \ccontext \s \mid \s\ccontext \mid \ccontext [\p/\s] \mid \s[\p/\ccontext] \\
\mbox{(\deft{Head Contexts})}    &  \hcontext & ::= & \Box   \mid \lambda \p. \hcontext \mid \hcontext \s  \mid  \hcontext [\p/\s] \\
  \end{array}  
  $$
\noindent where $\x,\y,\z$ range over a countable set of
variables, and every pattern $\p$ is {\it linear}, \ie\ every variable
appears at most once in $\p$. We denote by $\id$ the \deft{identity
  function} $\lambda \x. \x$ and by $\dup$ the auto applicative function $\lambda \x.\x\x$. As usual we use the abbreviation
$\lambda \p_1 \ldots \p_n. \s_1 \ldots \s_m $ for
$\lambda \p_1 (\ldots (\lambda \p_n. ((\s_1 \s_2) \ldots \s_m))\ldots)$,
$n\geq 0$, $m\geq 1$.  Remark that every $\lambda$-term is in particular a $\Lp$-term.


\noindent
The operator $[\p/\s]$ is called an \deft{explicit matching}. The
constant $\fail$ denotes the failure of the matching operation.
The sets of \deft{free} and \deft{bound} variables of a term $\s$, denoted respectively by $\fv{\s}$ and $\bv{\s}$, are defined as
expected, in particular $\fv{\lambda \p.\s} := \fv{\s} \setminus \fv{\p}$
and $\fv{\s[\p/\uu]} := (\fv{\s} \setminus \fv{\p}) \cup \fv{\uu}$.  A
term $\s$ is \deft{closed} if $\fv{\s} = \emptyset$. We write
$\p \# \q$ iff $\fv{\p} \cap \fv{\q} = \es$.  As usual, terms are
considered modulo $\alpha$-conversion.  Given a term
  (resp. list) context $\ccontext$ (resp. $\slist$) and a term $\s$,
  $\appctx{\ccontext}{\s}$ (resp. $\appctx{\slist}{\s}$) denotes the term
  obtained by replacing the unique occurrence of $\square$ in
  $\ccontext$ (resp. $\slist$) by $\s$, thus possibly capturing some free
  variables of $\s$. In this paper,
    an occurrence of a subterm $\uu$ in a term $\s$ is
    understood as the unique context $\ccontext$ such
    that $\s = \ccontext[\uu]$.

\medskip
The \deft{reduction relation} of the $\Lp$-calculus, denoted by $\Rew{}$, is the $\ccontext$-contextual closure of the following rewriting rules:
\[ \begin{array}{llllllllllllll}
  (\run) &    \appctx{\slist}{\lambda \p. \s} \uu  & \mapsto & \appctx{\slist}{\s[\p/\uu]}  \\
  (\rdeux) &    \s[\x/\uu]            & \mapsto & \s\isubs{\x}{\uu} \\
(\rtrois) &   \s[\pair{\p_1}{\p_2}/\appctx{\slist}{\pair{\uu_1}{\uu_2}}]  & \mapsto & \appctx{\slist}{\s[\p_1/\uu_1][\p_2/\uu_2]}  \\  
(\rquatre) &    \s[\pair{\p_1}{\p_2}/\appctx{\slist}{ \lambda \q. \uu}]  & \mapsto & \fail  \\
(\rhuit) &   \appctx{\slist}{ \pair{\s}{\uu}} \vv  & \mapsto & \fail  \\  
(\rcinq) &    \s[\pair{\p_1}{\p_2}/\fail ]  & \mapsto & \fail \\
(\rneuf) &   \appctx{\slist}{\fail}  & \mapsto & \fail  \\  
(\rsix) & \fail\  \s & \mapsto & \fail  \\
(\rsept) &    \lambda \p. \fail   & \mapsto & \fail  \\
\end{array} \]

\noindent where $\s\isubs{\x}{\uu} $ denotes the substitution of all
the free occurrences of $\x$ in $\s$ by $\uu$ and $\slist\neq\Box$ in
rule $(\rneuf)$.  By $\alpha$ -conversion, and \wlg, no reduction rule
captures free variables, so that in particular $\bv{\slist}\cap
\fv{\uu} = \es$ holds for rule $(\run)$ and $\bv{\slist} \cap \fv{\s}
= \es$ holds for rule $(\rtrois)$.  The rule $(\run)$ triggers the
pattern operation while rule $(\rdeux)$ performs substitution, rules
$(\rtrois)$ and $(\rquatre)$ implement (successful or unsuccessful)
pattern matching.  Rule $(\rhuit)$ prevents bad applications and rules
$(\rcinq)$, $(\rneuf)$, $(\rsix)$ and $(\rsept)$ deal with propagation
of failure in {\tt r}ight/{\tt l}eft {\tt e}xplicit {\tt m}atchings,
{\tt l}eft {\tt a}pplications and {\tt abs}tractions, respectively.
A \deft{redex} is a term having the form of the left-hand side
  of some rewriting rule $\mapsto$. The reflexive and transitive
closure of $\Rew{}$ is written $\Rewn{}$.

\begin{lemma} \label{l:confluence} The reduction relation $\Rew{}$ is confluent.
  \end{lemma}
\begin{proof} The proof is given in the next subsection. \end{proof} 

\medskip

\deft{Normal forms} ${\cal N}$ are terms without occurrences of redexes; 
  they are formally defined by the following grammars:
\[
 \begin{array}{lll}
    {\cal N} & ::= &  \fail \mid {\cal O} \\
    {\cal O} & ::= & {\cal M} \mid \lambda \p. {\cal O}  \mid  \pair{\cal N}{\cal N}  \mid   {\cal O} [\pair{\p}{\q}/{\cal M}] \\
    {\cal M}  & ::= & \x \mid {\cal M}{\cal O} \mid  {\cal M} [\pair{\p}{\q}/{\cal M}]  \\
  \end{array}    
\]
\begin{lemma} A term $\s$ is an ${\cal N}$-normal form if and only if
  $\s$ is a $\Rew{}$-normal form, \ie\ if no rewriting rule is applicable to
 any subterm of $\s$.
  \end{lemma}

We define a term to be  \deft{normalizing} if it reduces to a normal form. 

Let us notice that in
a language like $\Lp$, where there is an explicit notion of failure,
normalizing terms are not interesting from a computation point of
view, since $\fail$ is a normal form, but cannot be considered as the
result of a computation. If we want to formalize the notion of 
programs yielding a result, terms reducing to $\fail$ cannot be taken
into consideration. 
Remark however that $\pair{\fail}{\fail}$ is not operationally equivalent to $\fail$, according to  the idea that a pair can always be observed, and so it can
be considered as a result of a computation. This suggests a notion of
 reduction which is \deft{lazy w.r.t. pairs}, 
 \ie\  that never reduces  inside pairs. 
Indeed: 
   \begin{definition}
     A term $\s$ is  \deft{solvable} 
 if  there exists a
 head context $\hcontext$ such that $\appctx{\hcontext}{\s} $ is closed
and $\appctx{\hcontext}{\s} \Rewn{} \pair{\uu}{\vv}$, for some terms $\uu$ and $\vv$.
\end{definition}

\medskip

Therefore, a syntactical class of terms which is particularly interesting from an operational point of view is that of canonical forms. \deft{Canonical forms} ${\cal J}$ (resp. \deft{pure canonical forms $\cal J'$}) can be formalized by the following grammar:

\[
 \begin{array}{l@{\hspace{2cm}}l}
    {\cal J}  ::=  \lambda \p. {\cal J}  \mid  \pair{\s}{\s}  \mid {\cal K} \mid   {\cal J} [\pair{\p}{\q}/{\cal K}] &
    {\cal J'}  ::=  \lambda \p. {\cal J'}  \mid  \pair{\s}{\s}  \mid {\cal K'} \mid   {\cal J'} [\pair{\p}{\q}/{\cal K}']     \\
{\cal K}   ::=   \x  \mid {\cal K} \s \mid {\cal K} [\pair{\p}{\q}/{\cal K}] & 
{\cal K'}   ::=   \x  \mid {\cal K'} \s 
  \end{array}
\]
where the notion of pure canonical form, \ie\ of canonical form  without nested matchings, is a technical notion that will be useful in the sequel.
A term $\s$ \deft{is in \canonical} (or it is \deft{canonical}), 
written  ${\can}$,  if it  is generated by 
${\cal J}$, and it \deft{has a \canonical} if it reduces to a term
in $\can$. Note that ${\cal K}$-canonical forms cannot be closed.
Also, remark that the $\can$  of a term is not unique, \eg\ 
both $\pair{\id}{\id\ \id}$ and $\pair{\id}{\id}$ are $\can$s  of $(\lambda \x \y. \pair{\x}{\y})\ \id\ (\id\ \id)$.
It is worth noticing that $\cal N \cap \cal J \not= \emptyset$
  but  neither $\cal N \subset \cal J $ nor $\cal J \subset \cal N $. 
Latter, we will prove that 
solvable terms are strictly contained in the canonical ones.
\begin{example} \mbox{}
  \begin{itemize}
\item The term $\pair{\fail}{\fail}$ is both in normal and canonical form.
    \item The term $\fail$ is in normal form, but not in canonical form.
\item The term $\pair{\dup \dup}{\dup \dup}$ is in canonical form, but not in normal form.
\item The term $\lambda \pair{\x}{\y}. \id [\pair{\z_1}{\z_2}/\y \id [\pair{\y_{1}}{\y_{2}}/\z] ]$ is in canonical form, 
but not in pure canonical form.
\item The term $\lambda \pair{\x}{\y}. \id [\pair{\z_1}{\z_2}/\y \id]  $ is in pure canonical form.
\end{itemize}
\end{example}

We end this section by stating a lemma about $(\rdeux)$-reduction that
will be useful in next Section. 

\begin{lemma}\label{lem:inf} \mbox{}
Every infinite $\Rew{}$-reduction sequence
contains an infinite number of $(\rdeux)$-reduction steps.
\end{lemma}
\begin{proof} 
  It is sufficient to show that the reduction system without the rule
  $(\rdeux)$, that we call $\altp$, is terminating.  Indeed, remark
  that $\s \Rew{\altp} \s'$ implies $\nu(\s) >\nu(\s')$, where
  $\nu(\s)$ is a pair whose first component is the number of
  applications of the form $\appctx{\slist}{\uu_1} \uu_2$ in $\s$
  (rules $\run, \rsix, \rhuit$), and whose second component is the sum
  of the sizes of patterns in $\s$ (rules
  $\rtrois, \rquatre, \rcinq, \rsept, \rneuf$).  These
  pairs are ordered lexicographically.
\end{proof}

\subsection{The Confluence Proof}

\newcommand{\szero}{s_0}
\newcommand{\suno}{s_1}
\newcommand{\sdos}{s_2}
\newcommand{\stres}{s_{3}}
\newcommand{\scuatro}{s_4}
\newcommand{\scinco}{s_{5}}
\newcommand{\sseis}{s_6}
\newcommand{\ssiete}{s_7}
\newcommand{\socho}{s_8}
\newcommand{\snueve}{s_9}
\newcommand{\sregle}[1]{s_{#1}}

\medskip

In order to show confluence of our reduction system $\Rew{}$ we first
simplify the system by erasing just one rule $\szero$ in such a way
that confluence of $\Rew{}$ holds if confluence of $\Rew{}$ 
deprived from $\Rew{\szero}$ holds. This last statement is proved by applying the
decreasing diagram technique~\cite{vOdecreasing}. We just change the name/order
of the rules to make easier the application of the decreasing technique.

\[
   \begin{array}{lllllllll}
(\szero) &    \s[\x/\uu]            & \mapsto & \s\isubs{\x}{\uu}\\
(\suno) &    \appctx{\slist}{\fail}  & \mapsto & \fail\ & (\slist \mbox{ non-empty}) \\ 
(\sdos) &  \fail\  \s & \mapsto & \fail  \\
(\stres) &    \lambda \p. \fail   & \mapsto & \fail  \\
(\scuatro) &   \appctx{\slist}{ \pair{\s}{\uu}} \vv  & \mapsto & \fail  \\  
(\scinco) &    \s[\pair{\p_1}{\p_2}/\fail ]  & \mapsto & \fail \\
(\sseis) &    \s[\pair{\p_1}{\p_2}/ \appctx{\slist}{\lambda \q. \uu}]  & \mapsto & \fail  \\
(\ssiete) &    \s[\pair{\p_1}{\p_2}/\appctx{\slist}{\pair{\uu_1}{\uu_2}}]  & \mapsto &  \appctx{\slist}{\s[\p_1/\uu_1][\p_2/\uu_2]}  \\
(\socho) &  \appctx{\slist}{\lambda \p. \s} \uu & \mapsto &  \appctx{\slist}{\s[\p/\uu]} \\
    \end{array} 
\]

We define $\altp := \Rew{} \!  \setminus \szero$. 
We write $\s \Reweq{\szero} \s'$ iff 
$\s \Rew{\szero} \s'$ or $\s = \s'$.

\begin{lemma}
\label{l:concrete-local-commmutation}
For  all $\s_0,  \s_1, \s_2$,  if $\s_0  \Rew{\szero} \s_1$  and $\s_0
\Rew{\altp} \s_2$,  then there  exists $\s_3$ s.t.  $\s_1 \Rewn{\altp}
\s_3$ and $\s_2 \Reweq{\szero} \s_3$.
\end{lemma} 

\begin{proof}
By induction on $\s_0 \Rew{\szero} \s_1$, we only show the most significant cases:
\[ \begin{array}{clc@{\hspace{2cm}}clc}
    \s[\x/\uu]  & \Rew{\szero} & \s\isubs{\x}{\uu} &
    \s[\x/\uu]  & \Rew{\szero} & \s\isubs{\x}{\uu} \\
    \mbox{}_{\altp}\downarrow{} & & \mbox{}_{\altp}\downarrow{} &
    \mbox{}_{\altp}\downarrow{} & & \mbox{}_{\altp}\downarrow_* \\
    \s' [\x/\uu]  & \Rew{\szero} & \s'\isubs{\x}{\uu} &
  \s [\x/\uu']  & \Rew{\szero} & \s\isubs{\x}{\uu'} \\ \\
 \end{array}\]

\[\begin{array}{cccc}
\appctx{\slist_2}{\appctx{\slist_1}{\fail} [\x/\uu]}   & \Rew{\szero} &  \appctx{\slist_2}{\appctx{\slist_1}{\fail}  \isubs{\x}{\uu}}   \\
 \mbox{}_{\suno}\downarrow & &  \mbox{}_{\suno}\downarrow_=   \\
     \fail  & = & \fail
  \end{array}\]
\par \vspace{-\baselineskip}
\qedhere
\end{proof}

The  following lemma can be found for example in~\cite{barendregt84nh}.

\begin{lemma}
\label{l:local-commutation}
Let $\Rew{\R_1}$ and $\Rew{\R_2}$ be two reduction relations. Suppose 
for any $\s_0, \s_1, \s_2$ such that $\s_0 \Rew{\R_2} \s_1$ and $\s_0 \Rew{\R_1} \s_2$, 
there exists $\s_3$ verifying $\s_1 \Rewn{\R_1} \s_3$ and $\s_2 \Reweq{\R_2} \s_3$.
Then $\Rew{\R_2}$ and $\Rew{\R_1}$ {\bf commute}, \ie\ 
$\forall \s_0, \s_1, \s_2$ if $\s_0 \Rewn{\R_2} \s_1$ and $\s_0 \Rewn{\R_1} \s_2$, 
$\exists \s_3$ s.t. $\s_1 \Rewn{\R_1} \s_3$ and $\s_2 \Rewn{\R_2} \s_3$.
\end{lemma}

By Lemma~\ref{l:local-commutation} and~\ref{l:concrete-local-commmutation} we obtain:

\begin{corollary}
\label{l:commutation}
The reduction relations $\Rew{\altp}$ and $\Rew{\szero}$ commute. 
\end{corollary}

\begin{lemma}
\label{l:alt-confluent}
The reduction relation $\Rew{\altp}$ is confluent.
\end{lemma}

\begin{proof} 
 We use the decreasing diagram technique~\cite{vOdecreasing}.
For that, we first order the reduction rules of the system $\altp$
by letting $\sregle{i} < \sregle{j}$ iff $i < j$. We write
$\s \Rew{i} \uu$ if $\s \Rew{\sregle{i}} \uu$. Given a set $\I$ of natural numbers we
write  $\s \Rewn{\borne \I} \uu$ 
if every $\Rew{j}$-reduction step  in the sequence $\s \Rewn{\borne  \I} \uu$
verifies $j < \I$, \ie\ if for every $\Rew{j}$-reduction step 
in the sequence  $\exists i \in \I$ such that $j < i$. 
The system $\Rew{\altp}$ is
said to be decreasing iff for any $\s_0, \s_1, \s_2$ 
such that $\s_0 \Rew{l} \s_1$  and $\s_0 \Rew{m} \s_2$,
there exists $\s_3$ such that $\s_1 \Rewn{\borne  \set{l}} \Reweq{m} \Rewn{\borne  \set{l,m}} \s_3 $
and $\s_2 \Rewn{\borne  \set{m}} \Reweq{l} \Rewn{\borne  \set{l,m}} \s_3 $,
where $\s \Reweq{l} \s'$ means $\s \Rew{l} \s'$ or $\s = \s'$.

We now  show that the system $\Rew{\altp}$ is decreasing.
As a matter of notation, we write for example  $\s \Rew{3=, 5}
\s'$ to denote a rewriting sequence
of length 2 or 1, composed respectively   by a $\Rew{3}$-step followed by a  $\Rew{5}$-step or
  by a single  $\Rew{5}$-step.
\begin{itemize}
\item We consider the cases $\s_0  \Rew{1} \s_1$ and $\s_0 \Rew{i} \s_2\ (i = 1 \ldots 8)$. 
We only show the interesting ones. 

\[  \begin{array}{clccc}
\appctx{\slist_2}{\appctx{\slist_1}{\fail}[\pair{\p_1}{\p_2}/\fail]}   & \Rew{1} & \fail  \\
    \mbox{}_{5}\downarrow{} & & =  \\
   \appctx{\slist_2}{\fail}  &  \Reweq{1}  & \fail \\  \\
    \end{array} \]
  
\[     \begin{array}{clccc}
 \appctx{\slist_2}{\appctx{\slist_1}{\fail}[\pair{\p_1}{\p_2}/\appctx{\slist}{\pair{\uu_1}{\uu_2}}]}  & \Rew{1} & \fail \\
    \mbox{}_{7}\downarrow{} & & = \\
   \appctx{\slist_2}{ \appctx{\slist}{\appctx{\slist_1}{\fail}[\p_1/\uu_1][\p_2/\uu_2]}}  & \Reweq{1} & \fail \\  
     \end{array} \] 

All of them are decreasing diagrams as required.

\item The cases $\s_0  \Rew{2} \s_1$ and $\s_0 \Rew{i} \s_2\ (i = 2 \ldots 8)$
are straightforward.

\item The interesting case $\s_0 \Rew{3} \s_1$ and $\s_0
  \Rew{i} \s_2\ (i = 3 \ldots 8)$ is the following.
\[
    \begin{array}{clc}
   \appctx{\slist}{\lambda \p. \fail} \uu & \Rew{3} &\appctx{\slist}{ \fail} \uu  \\
    \mbox{}_{8}\downarrow{} & & \mbox{}_{1=,2 }\downarrow{} \\
    \appctx{\slist}{\fail[\p/ \uu]} & \Rew{1} & \fail \\
     \end{array}
\]
\item The interesting cases $\s_0 \Rew{4} \s_1$ and $\s_0
  \Rew{i} \s_2\ (i = 4 \ldots 8)$ are straightforward. 

\item The interesting cases $\s_0 \Rew{5} \s_1$ and $\s_0
  \Rew{i} \s_2\ (i = 5 \ldots 8)$ are the following.
\[
    \begin{array}{clc}
      \s[\pair{\p_1}{\p_2}/\appctx{\slist_1}{\appctx{\slist_2}{\pair{\uu_1}{\uu_2}}[\pair{\q_1}{\q_2}/\fail]}] & \Rew{5} &
      \s[\pair{\p_1}{\p_2}/\appctx{\slist_1}{\fail}] \\
    \mbox{}_{7}\downarrow{} & & \mbox{}_{1^=,5}\downarrow{} \\
   \appctx{\slist_1}{\appctx{\slist_2}{\s[\p_1/\uu_1][\p_2/\uu_2]}[\pair{\q_1}{\q_2}/\fail]} & \Rew{5} &  \appctx{\slist_1}{\fail}\\ \\
     \end{array}
\]

\[
    \begin{array}{clc}
      \appctx{\slist_1}{ \appctx{\slist_2}{\lambda \p. \s}[\pair{\q_1}{\q_2}/\fail]} \uu  & \Rew{5} & \appctx{\slist_1}{ \fail} \uu\\
    \mbox{}_{8}\downarrow{} & & \mbox{}_{1^=,2}\downarrow{} \\
   \appctx{\slist_1}{\appctx{\slist_2}{\s[\p/\uu]}[\pair{\q_1}{\q_2}/\fail]} & \Rew{5,1=} &  \fail \\
     \end{array}
\]
\item The cases for  $\s_0 \Rew{6} \s_1$ and $\s_0
  \Rew{\sregle{i}} \s_2\ (i = 6 \ldots 8)$ are similar. 

\item The cases for  $\s_0 \Rew{7} \s_1$ and $\s_0
  \Rew{\sregle{i}} \s_2\ (i = 7 \ldots 8)$ have the following
  reduction scheme:
  
\[
    \begin{array}{clc@{\hspace{1.5cm}}clc}
      \cdot  & \Rew{7} & \cdot & \cdot  & \Rew{7} & \cdot \\
      \mbox{}_{7}\downarrow{} & & \mbox{}_{7}\downarrow{}
      & \mbox{}_{8}\downarrow{} & & \mbox{}_{8}\downarrow{} \\
  \cdot  & \Rew{7} & \cdot & \cdot  & \Rew{7} & \cdot \\
     \end{array}
\]

  \item There is no other case.  \qedhere

\end{itemize}

\end{proof}

\begin{lemma}[Hindley-Rosen]
\label{l:confluence-from-commutation}
Let  $\Rew{\R_1}$ and $\Rew{\R_2}$ be two confluent reduction relations 
which commute. Then $\Rew{\R_1 \cup \R_2}$ is confluent.  
\end{lemma}

\noindent {\bf Lemma~\ref{l:confluence}.}
The reduction relation $\Rew{}$ is confluent.

\begin{proof}
Since $\Rew{\szero}$ is trivially confluent,
$\Rew{\altp}$ is confluent by Lemma~\ref{l:alt-confluent}
and $\Rew{\altp}$ and $\Rew{\szero}$ commute by Corollary~\ref{l:commutation},
then $\Rew{\altp} \cup \Rew{\szero} = \Rew{}$ turns out to be confluent 
by Lemma~\ref{l:confluence-from-commutation}, 
which concludes the proof. 
\end{proof}


\section{The Type System $\Pu$}
\label{s:type-system}

In this section we present a type system for the $\Lp$-calculus, and
we show that it characterizes terms having \canonical, \ie\ that a
  term $\s$ is typable if and only if $\s$ has canonical form.  \deft{The
  set ${\mathcal T}$ of types} is generated by the following grammar:
\[ \hfill \begin{array}{llll}
     \mbox{(Types)} &  \sig, \tau & ::= & \alpha  \mid \PT \mid \A \rew \sig \\
     \mbox{(Product Types)} & \PT & :: = & \prodt{\A}{\D} \\
     \mbox{(Multiset Types)} & \A, \D & :: = & \mult{\sig_k}_{\kK} \\
   \end{array} \hfill     \]

where $\alpha$ ranges over a countable set of constants, $K$
is a (possibly empty) finite set of indices, and a multiset is an
unordered list of (not necessarily different) elements. The arrow type
constructor $\arrow$ is right associative.

\deft{Typing environments}, written $\Gamma,
\Delta, \Lambda$, are functions from variables to multiset types, assigning the
empty multiset to almost all the variables. The \deft{domain} of $\Gamma$,
written $\dom{\Gam}$, is the set of variables whose image is different
from $\emul$.  We may write $\Gam \# \Del$ iff 
$\dom{\Gam} \cap \dom{\Del} = \es$. 

\begin{notation}
Sometimes we will use symbols $\mu,\nu$ to range over the union of
types and multiset types. We abbreviate by the constant $\oprod$ the
product type $\prodt{\emul}{\emul}$.
We write $\munion$ to denote multiset union and $\sqsubseteq$ multiset
inclusion; these operations take multiplicities into account.
Moreover, abusing the notation we will use $\in$ to denote both set
and multiset membership.

Given typing environments $\{\Gamma_i\}_{\iI}$, we write
$+_{\iI}\Gamma_i$ for the environment which maps $\x$ to
$\munion_{\iI}\Gamma_i(\x)$. If $I=\es$, the resulting environment is the one having an empty
domain. Note that $\Gamma +\Delta$ and $\Gamma +_{\iI} \Delta_i$ are
just particular cases of the previous general definition.  When $\Gam
\# \Del$ we may write $\Gamma; \Delta$ instead of $\Gamma +\Delta$.
The notation  $\Gamma\sm \Delta$  is used  for the environment 
whose domain is $\dom{\Gamma} \sm \dom{\Delta}$, defined 
as expected;  ${\x}_1\!:\!{\A}_1;\ldots ; {\x}_n\!:\!{\A}_n$
is the environment assigning ${\A}_i$ to ${\x}_i$, for $1\leq i\leq
n$, and $\emul$ to any other variable; 
$\Gam|_{\p}$ denotes the environment such that 
$\Gam|_{\p}(\x)=\Gam(\x)$, if $\x \in \fv{\p}$,
$\emul$ otherwise. We also assume that $\Gam; \x:\emul$ is identical to $\Gam$.
Finally, $\Gam \sqsubseteq \Del$ means that
  $\x \in \dom{\Gam}$ implies
  $\x \in \dom{\Del}$
  and $\Gam(\x) \sqsubseteq \Del(\x)$.
\end{notation}

The \deft{type assignment system} $\Pu$
(\cf\ Figure~\ref{fig:system-pattern}) is a set of typing rules
assigning both types and multiset types of ${\mathcal T}$ to terms of
$\Lp$.  It uses an auxiliary system assigning multiset types to
patterns.  We write $\Pi \dem \Gam \der \s:\sig$ (resp.  $\Pi \dem
\Gam \der \s:\A$ and $\Pi \dem \Gam \pder \p:\A$) to denote a
\deft{typing derivation} ending in the sequent $\Gam \der \s:\sig$
(resp. $\Gam \der \s:\A$ and $\Gam \pder \p:\A$), in which case $\s$
(resp.  $\p$) is called the {\bf subject} of $\Pi$ and $\sig$ or $\A$
its {\bf object}. By abuse of notation, $\Gam \der \s:\sig$
(resp. $\Gam \der \s:\A$ and $\Gam \pder \p:\A$) also denotes the
existence of some typing derivation ending with this sequent. A
derivation $\Pi \dem \Gam \der \s:\mu$ is \deft{meaningful} if $\mu$
is either a type or a multiset $\not= \emul$.  A
    pattern $\p$ is {\bf typable} if there is a derivation whose
    subject is $\p$; a term $\s$ is {\bf typable} if there is
    derivation whose subject is $\s$ and whose object is a type, or
    equivalently if there is a a meaningful derivation whose subject
    is $\s$.  We will prove later that every pattern is in fact
    typable (\cf\ Corollary~\ref{cor:gen-pair}).  The \deft{measure}
    of a typing derivation $\Pi$, written $\meas(\Pi)$, is the number
    of all the typing rules in $\Pi$ except $(\many)$.

\begin{figure}[t]
\begin{framed}
\begin{center}
\[ \hfill \begin{array}{c}
       \infer{}{\x:\A  \pder \x:\A}\ (\trvarpat) \sep \sep 
      \infer{\Gam \pder \p: \A \sep \Delta \pder \q:\D \sep \p \# \q }
        { \Gam +  \Delta \pder \pair{\p}{\q} :
        \mult{\prodt{\A}{\D}} }\ (\trpairpat)  \\\\
     {\bf Patterns} \\
     \hline \\
        \infer{  } 
      {\seq{\x: \mult{\sig}}{  \x: \sig} }\ (\ax)    \sep \sep \sep
      \infer{(\seq{\Gamk}{\s: \sigk})_{\kK}}
        {\seq{+_{\kK}\Gamk}{\s: \mult{\sigk}_{\kK}}}\ (\many) \\ \\ 
        \infer{ \seq{\Gam}{\s : \sig}\sep        \Gam|_{\p} \pder  \p:\A}  
       {\seq{\Gam\sm \Gam|_\p}{\lambda \p.\s: \A \rew  \sig}}\ (\introarrow) \sep \sep       
 \infer{ \seq{\Gam}{\s: \A \rew \sig} \sep\sep 
        \seq{\Del}{\uu : \A }}
        {\seq{\Gam + \Del}{ \s\,\uu: \sig}}\ (\app)      \\ \\ 
 \infer{\seq{\Gam}{\s:\A} \sep
        \seq{\Del}{\uu:\D}}
             {\seq{\Gam +\Del}{\pair{\s}{\uu} : \prodt{\A}{\D}}}\ (\trpair) \sep\sep 
        \infer{\seq{\Gam}{\s:\sig} \sep
        \Gam|_{\p} \pder \p:\A \sep
        \seq{\Del}{\uu:\A}}
        {\seq{(\Gam\sm \Gam|_\p) +  \Del}{\s[\p/\uu]:\sig}}\ (\trsub) \\ \\
        {\bf Terms} \\
            \hline
          \end{array} 
     \hfill \]  
\end{center}
\caption{The Type Assignment System $\Pu$.}
\label{fig:system-pattern}
\end{framed}
\end{figure}

Rules $(\ax)$ and
$(\app)$ are those used for the $\lambda$-calculus
in~\cite{bkdlr14,deC09}. Rule $(\introarrow)$
is the natural extension to patterns of the standard
rule for abstraction used in the $\lambda$-calculus. 
Linearity of patterns is guaranteed by the
clause $\p \# \q$ in rule $(\trpairpat)$. Rule $(\trvarpat)$ with
$\A = \emul$ is essential to type {\it erasing} functions such as for
example $\lambda \x. \id$.  The rule $(\many)$ allows to assign
  multiset types to terms, and it cannot be iterated; in particular
  note that every term can be assigned with this typing rule the empty
  multiset type by setting $K=\emptyset$.  The rule $(\trpair)$ is self
  explanatory; note that, in case both the objects of its premises are
  empty multisets, it allows to type a pair like
  $\pair{\dup \dup}{\dup \dup}$ without assigning types to its
  components, thus
  $(\lambda \pair{\x}{\y}. \id) \pair{\dup \dup}{\dup \dup}$ will be
  typable too, whereas $\dup \dup$ will not.  This choice reflects
the fact that every pair is canonical, so any kind of pair needs to be
typed.
Rule ($\trsub$) is the more subtle one: in order to type
$ \s[\p/\uu]$, on one hand we need to type $\s$, and on the other one,
we need to check that $\p$ and $\uu$ can be assigned the same
types.

The system is relevant, in the sense that only the used premises are registered in the typing environments. This
property, formally stated in the following lemma, will be an important
technical tool used to develop the inhabitation algorithm.

\begin{lemma}[Relevance]
\label{l:relevance} \mbox{}
\begin{itemize}
\item\label{l:relevance-uno} If $\Gam \pder \p:\A$, then $\dom{\Gam} \subseteq \fv{\p}$.
\item\label{l:relevance-due} If $\Gam \der \s:\sig $, then $\dom{\Gam} \subseteq \fv{\s}$.
\end{itemize}
\end{lemma}

\begin{proof}
By induction on the typing derivations. 
\end{proof}

A first elementary property of
the type system is that head occurrences are always typed:

\begin{lemma}
  \label{l:head-contexts}  If $\appctx{\hcontext}{\s}$ is typable, then $t$ is typable.
\end{lemma}
\begin{proof}
  Let $\Pi \dem \Gam \der \appctx{\hcontext}{\s}: \mu$ be meaningful, then it is easy to prove, by induction
 on $\hcontext$, that the occurrence of $\s$ filling the hole of $\hcontext$ is always typed.
\end{proof}

\subsection{On the Typing of Patterns}
The system $\Pu$ features two kinds of typings: those of the form
$\Gam\der\s:\mu$, for terms, and those of the form
$\Gam\pder\p:\A$, for patterns.  These are, of course, fundamentally
dissymetric notions: $\der$ is undecidable and non deterministic (a
given term may have several types in a given environment), whereas
$\pder$ is decidable and deterministic.  As a matter of fact  the {\it unique} type $\A$ such that
$\Gam\pder\p:\A$ could have been denoted by $\Gam(\p)$, as 
we do for
variables.  However, we decided to keep the typing judgements $\Gam
  \pder \p:\A$ since they allow for a clearer formulation of
the typing rules of $\Pu$.
Some preliminary definitions are given below
to prove the uniqueness of the typing of patterns.

Given two patterns $\p$ and $\q$, we say that $\p$ \deft{occurs} in  $\q$ if
\begin{itemize}
\item either $\p=\q$
  \item or
    $\q=\pair {\q_1} {\q_2}$ and either  $\p$ occurs  in $\q_1$ or  $\p$ occurs occurs in $\q_2$.
    \end{itemize}
  Remark that, by linearity of patterns, at most one of the  conditions in the second item above  may hold.

  If $\p$ occurs in $\q$, the multiset type $\A_\p^\q$ is defined as follows:
\begin{itemize}
   \item  $ \A_\p^\p  =\A$
   \item
$ \mult{\prodt{\A}{\D}}_\p^{\pair {\q_1} {\q_2}}  = \left\{\begin{array}{ll}
\A_\p^{\q_1} & \mbox{if $\p$ occurs in $\q_1$}\\
\D_\p^{\q_2} & \mbox{if $\p$ occurs in $\q_2$}\\
\end{array} \right .$
 \item $ \A_\p^{\pair {\q_1} {\q_2}} $ is undefined if $\A$ is not of the shape $ \mult{\prodt{\C}{\D}}$, for some $\C,\D$.
        \end{itemize}

Typings of patterns can  be characterized as follows:

\begin{lemma}\label{l:patterns}
For every  environment $\Gamma$ and pattern $\p$,   $\Gamma \pder \p:  \A$ if and only if  $\dom{\Gam} \subseteq \fv{\p}$ and for all $\q$ occurring in $\p$, $\A^\p_\q$ is defined and $\Gamma|_\q \pder \q:\A^\p_\q$.
\end{lemma}

\begin{proof}

  ($\Rightarrow$):
  If  $\Gamma \pder \p:  \A$, then   $\dom{\Gam} \subseteq \fv{\p}$ by Lemma~\ref{l:relevance}. The proof is
by  induction on $\p$. If $\p=\x$ then either $\x\not\in \dom{\Gamma}$ and
  $\A=\emul$, or $\Gamma(\x)=\D\not=\emul$ and $\A=\D^\x_\x=\D$.
  If the considered pattern is $\pair {\p}{\q}$, then the last rule of its type derivation is
  
  $$\infer{\Gam \pder \p: \D \quadt
\Del\pder \q:\C \quadt
\p \# \q}
{ \Gam + \Del \pder \pair{\p}{\q} :   \mult{\prodt{\D} {\C} }}
$$
where $\A= \mult{\prodt{\D} {\C} }$.
By the induction hypothesis $\dom{\Gam} \subseteq \fv{\p}$, $\dom{\Del} \subseteq \fv{\q}$ and
for every $\p' ,\q'$ occurring respectively in $\p$ and $\q$, $\Gam|_{\p'}\pder \p': \D_{\p'}^{\p}$ and
 $\Del|_{\q'}\pder \q': \D_{\q'}^{\q}$.
 Note that every pattern occurring in $\pair{\p}{\q}$ is either $\pair{\p}{\q}$ itself or it occurs 
 in exactly  one of the two components of the pair. In the first case the proof is trivial.
In the second one, since the domains of $\Gam$ and $\Del$ are
disjoint, by linearity of $\pair{\p}{\q}$,  $\Gam|_{\p'}+ \Del|_{\p'}=(\Gam +\Del)|_{\p'}$
and $\Gam|_{\q'}+ \Del|_{\q'}=(\Gam +\Del)|_{\q'}$,
and the proof follows by induction.

($\Leftarrow$): The proof is again by induction on $\p$.
  \end{proof}

  \begin{example} \mbox{} Let $\p=\pair {\pair \x \y } \w$ and
    $\Gamma=\x:\mult{\alpha,\beta},\y:\mult{\gamma}$. The (sub)patterns
    occurring in $\p$ are $\x$, $\y$, $\w$, $\pair \x \y$ and $\p$. In the
    typing 
    environment $\Gam$ restricted to its free variables, each (sub)pattern can be
    typed by a unique multiset:
 \[ \begin{array}{lll@{\hspace{1.5cm}}lll}
\Gam|_\x & = & \x:\mult{\alpha,\beta} \pder \x:\mult{\alpha,\beta} &
\Gam|_{\pair \x \y } & = & \Gam \pder {\pair \x \y }: \mult{\prodt{\mult{\alpha,\beta}}{\mult{\gamma}}} \\
\Gam|_\y & = & \y:\mult{\gamma} \pder \y:\mult{\gamma} &
\Gam|_ \p & =& \Gam \pder \p: \mult{\prodt{\mult{\prodt{\mult{\alpha,\beta}}{\mult{\gamma}}}}{\emul}} \\                                               
\Gam|_\w & = & \emptyset\pder \w:\emul \\
    \end{array} \] 
If we  denote  by $\A$ the type $ \mult{\prodt{\mult{\prodt{\mult{\alpha,\beta}}{\mult{\gamma}}}}{\emul}}$, then it is easy to verify that:
 \[ \begin{array}{lll@{\hspace{1.5cm}}lll}
\A^\p_\x & = &  \mult{\alpha,\beta} & \A^{\p}_{\pair \x \y }& = & \mult{\prodt{\mult{\alpha,\beta}}{\mult{\gamma}}} \\
\A^\p_\y & = & \mult{\gamma} & \A^\p_\p &= & \A   \\ 
\A^\p_\w & = &  \emul \\
 \end{array} \] 
 \end{example}
 The following corollary follows immediately. 
\begin{corollary}\label{cor:gen-pair}\mbox{}
  For every pattern $\p$ and every environment  $\Gam$ such that $\dom{\Gam} \subseteq \fv{\p}$, there exists a unique multiset $\A$ such that $\Gam \pder \p:\A$.
  Moreover if $\p=\pair{\p_1}{\p_2}$  then  $\A=\mult{\prodt{\D}{\C}}$, for some $\D,\C$.

\end{corollary}


\subsection{Main Properties of system $\Pu$}
We are going to define the notion of {\it typed occurrences} of a
typing derivation, which plays an essential role in the rest of this
paper: indeed, thanks to the use of non-idempotent intersection types,
a combinatorial argument based on a measure on typing derivations
(\cf\ Lemma~\ref{lem:red:exp}(\ref{lem:reduction})), allows to prove
the termination of reduction of redexes occurring in typed occurrences
of their respective typing derivations.

Given a typing derivation $\Pi\dem
\Gam\vdash \ccontext[\uu]:\sig$,
the occurrence of $\uu$
in the hole of $\ccontext$ is a 
typed occurrence of $\Pi$
if and only if $\uu$  is the
subject of a  meaningful subderivation of $\Pi$. More precisely:

\begin {definition}\label{def:occ}
Given a type derivation $\Pi$, the set of
\deft{typed occurrences} of $\Pi$, written $\toc{\Pi}$,
is the set of  contexts defined by induction on $\Pi$ as follows. 
\begin{itemize}
\item If $\Pi$ ends with $(\ax)$,   then 
$\toc{\Pi} := \set{\Box}$.
\item If $\Pi$ ends with $(\trpair)$ with subject $\pair{\uu}{\vv}$ and premises $\Pi_i$ ($i=\{1,2\}$) then \\
$\toc{\Pi} := \set{\Box} \cup \set{ \pair{\ccontext}{\vv}\mid \ccontext \in \toc{\Pi_1} }  \cup \set{ \pair{\uu}{\ccontext}\mid \ccontext \in \toc{\Pi_2} }$.
\item If $\Pi$ ends with $(\introarrow)$ 
 with subject $\lambda \p.\uu$ and 
premise $\Pi'$  then \\
$\toc{\Pi} := \set{\Box} \cup \set{ \lambda\p.\ccontext\mid \ccontext \in \toc{\Pi'}}$ .
\item If $\Pi$ ends with $(\app)$
 with subject $\s\uu$ and 
premises $\Pi_1$ and $\Pi_2$
with subjects  $\s$
and $\uu$ respectively, then 
$\toc{\Pi} := \set{\Box} \cup  \set{\ccontext \uu \mid \ccontext \in \toc{\Pi_1}
  }\cup \set{\s \ccontext\mid \ccontext \in \toc{\Pi_2}}$. 
\item  If $\Pi$ ends with $(\trsub)$ with subject $\s[\p/\uu]$ and
  premises  $\Pi_1$ and $\Pi_2$
with subjects $\s$ and  $\uu$ respectively, then $\toc{\Pi} := \set{\Box} \cup
  \set{\ccontext[\p/\uu]\mid \ccontext \in \toc{\Pi_1}} \cup 
  \set{\s[\p/\ccontext] \mid \ccontext  \in \toc{\Pi_2}}$.
  \item If $\Pi$ ends with $(\many)$, with premises $\Pi_k\ (k \in K)$, then
  $\toc{\Pi} := \cup_{\kK } \toc{\Pi_k}$.
\end{itemize}
\end{definition}

\begin{example}
Given the following derivations $\Pi$ and $\Pi'$, 
the occurrences $\Box$ and $\Box \y$ belong to both $\toc{\Pi}$ and $\toc{\Pi'}$ while
$\x \Box$ belongs to $\toc{\Pi}$ but not to $\toc{\Pi'}$.

\[ \begin{array}{c}
\Pi \dem 
   \infer{\infer{\mbox{}}{\x: \mult{\mult{\tau} \arrow  \tau} \der \x:  \mult{\tau} \arrow  \tau} \sep
         \infer{\mbox{}}{\infer{\y: \mult{\tau} \der \y: \tau}{\y: \mult{\tau} \der \y: \mult{\tau}}}}
     {\x: \mult{\mult{\tau} \arrow  \tau}, \y: \mult{\tau} \der \x \y: \tau } \\ \\
     \Pi' \dem 
   \infer{\infer{\mbox{}}{\x: \mult{\emul \arrow \tau} \der \x: \emul \arrow \tau} \sep \infer{}{\es \der \y: \emul} }
         {\x: \mult{\emul \arrow \tau}  \der \x \y: \tau }
     \end{array} \]
\end{example}

The type assignment system $\Pu$ enjoys the fundamental properties of
subject reduction and subject expansion, based respectively on
substitution and anti-substitution properties whose proofs
can be found in~\cite{Alves}.
  \begin{lemma}[{\bf     Substitution/Anti-Substitution Lemma}]
    \label{l:substitution-lemma} \mbox{}
\begin{enumerate}
\item \label{l:substitution} If  $\Pi_\s \dem \Gam_\s ; \x : \A \der \s : \tau$
and $\Pi_\uu \dem \Gam_\uu   \der \uu : \A$, then
there exists $\Pi \dem \Gam_\s + \Gam_\uu  \dem  \s\isubs{\x}{\uu} : \tau$
such 
that $\meas(\Pi)\leq\meas(\Pi_\s)+\meas(\Pi_\uu)$.
\item \label{l:antisubstitution} If  $\Pi \dem \Gam \der  \s\isubs{\x}{\uu} : \tau$, then there
  exist derivations $\Pi_t$ and $\Pi_\uu$, environments 
  $\Gam_\s, \Gam_\uu$ and multiset $\A$ such that
  $\Pi_\s \dem \Gam_\s ; \x : \A \der \s : \tau$,
  $\Pi_\uu \dem \Gam_\uu \der \uu:\A$ and $\Gam =
  \Gam_\s + \Gam_\uu$. 
\end{enumerate}
\end{lemma}

On the other hand, the measure of 
any typing derivation is not increasing by reduction.  Moreover, the
measure strictly decreases for the reduction steps that are
{\it typed}.  This property  makes easier the proof of the ``only
  if'' part of Theorem~\ref{l:characterization-canonical}.
\begin{lemma} \label{lem:red:exp} \mbox{}
\begin{enumerate}
\item {\bf (Weighted Subject Reduction)}
\label{lem:reduction}
If $\Pi \dem\Gam \der \s:\tau$ and $\s\Rew{} \s'$, then $\Pi'\dem \Gam
\der \s':\tau$ and $\meas(\Pi') \leq \meas(\Pi)$. Moreover, if the
reduced redex occurs in a typed occurrence of $\Pi$,  then $\meas(\Pi') <
\meas(\Pi)$.
\item {\bf (Subject Expansion)}
\label{lem:subexp}
If $\Pi' \dem \Gam' \der \s':\sig$ and $\s \Rew{} \s'$, then  $\Pi \dem \Gam \der \s:\sig$.
\end{enumerate}
\end{lemma}

\begin{proof}
Both proofs are by induction on the reduction relation
    $\s \Rew{} \s'$. For the base cases (reduction at the
    root position): the rules $(\run), (\rdeux), (\rtrois)$ are treated
    exactly as in~\cite{Alves}, and the rules 
    $(\rquatre),(\rhuit),(\rcinq),(\rneuf),(\rsix),(\rsept)$
    do not apply since the term $\s$ (resp. $\s'$) would not be typable.
    All the inductive cases are straightforward, and in particular,
    when the reduction occurs in an untyped position, then
    the measures of $\Pi$ and $\Pi'$ are equal.
 \end{proof}

Given $\Pi \dem \Gam \der \s: \tau$, the term $\s$ is said to be in
\deft{$\Pi$-normal form}, also written \deft{$\Pi$-nf}, if for every
typed occurrence $\ccontext \in \toc{\Pi}$ such that $\s =
\appctx{\ccontext}{\uu}$, the subterm $\uu$ is not a redex.

We are now ready to provide the logical characterization of terms
having \canonical. This proof has been already given in~\cite{Alves}, based on the property that, if a term has a canonical
  form, then there is a \deft{head} reduction strategy reaching
  this normal form. But in
  order to make this paper self contained, we reformulate here the
  proof, using (for the completeness proof) a more general approach,
  namely that every reduction strategy choosing at least all the typed
  redexes reaches a canonical form. 
\begin{theorem}[Characterizing Canonicity]
\label{l:characterization-canonical}
A term $\s$ is typable iff $\s$ has a canonical form. 
\end{theorem}
\begin{proof} \mbox{}
  
\begin{itemize} 
\item(if) We reason by induction on the grammar defining the canonical forms.

  We first consider ${\cal K}$-canonical forms, for which we prove a
  stronger property, namely that for every $ \s \in {\cal K}$, for
  every $\sig$ there is an environment  $\Gam$
   such that $\Gam \der
  \s:\sig$. We reason by induction on the grammar defining ${\cal K}$.

If $\s = \x$, the proof is straightforward. If $\s=\vv\uu$, then by
the \ih\ there is a typing derivation
$\Gam \der \vv:\emul \arrow \sig$. Since to every term the
   multiset $\emul$ can be assigned by rule ($\many$) with an empty
  premise, the result follows by application of rule
$(\app)$.

 Let $\s = \uu[\pair{\p}{\q}/\vv]$. By the \ih\ for every $\sigma$,
 there is $\Gam$ such that $\Gam \der \uu: \sig$. By
Corollary~\ref{cor:gen-pair}, $\Gam|_{\pair{\p}{\q}} \pder
 \pair{\p}{\q}:\mult{\prodt{\A}{\D}}$ for some $\A,\D$.
 Then   by the \ih\ again there is $\Delta$
 such that $\Delta\der
 \vv:\prodt{\A}{\D}$, and so, by rule $(\many)$, $\Delta \der
 \vv:\mult{\prodt{\A}{\D}}$. Then, by applying rule $(\trsub)$, we get $\Gam\sm
 \Gam|_{\pair{\p}{\q}} + \Delta \der
 \uu[\pair{\p}{\q}/\vv]: \sig$.

Now, let $\s$ be a ${\cal J}$-canonical form. If
  $\s=\pair{\uu}{\vv}$ then by rules ($\many$) and ($\trpair$)
  $\der \pair{\uu}{\vv}:\prodt{\emul}{\emul}$.  If $\s=\lambda\p.\uu$,
then $\uu$ can be typed by the \ih\ so that let $\Gam \der
\uu:\sig$. If $\p=\x$, for some $\x$, then by applying rule
$(\introarrow)$,
$\Gam\sm \Gam|_\x \der \lambda \x.\uu:\Gam(\x) \arrow \sig$,
otherwise, by Corollary~\ref{cor:gen-pair},
  $\Gam|_{\p}\pder \p:\mult{\prodt{\A}{\D}}$, for some $\A,\D$,
  and then
  $\Gam\sm \Gam|_\p \der \lambda\x.\uu:\mult{\prodt{\A}{\D}} \arrow
  \sig$, always by rule $(\introarrow)$. 
Let
$\s=\s'[\pair{\p}{\q}/\vv]$,
where $\s'$ (resp. $\vv$) is a ${\cal J}$ (resp. ${\cal K}$) canonical form. By the \ih\ there are
$\Gam,\sigma$ such that $\Gam \der \s':\sigma$. 
Moreover, Corollary~\ref{cor:gen-pair} gives
$\Gam|_{\pair{\p}{\q}}
\pder \pair{\p}{\q}:\mult{\prodt{\A}{\D}}$ for some
$\A,\D$. Since $\vv$ is a ${\cal K}$-canonical
form, then $\Del\der \vv:\prodt{\A}{\D}$ as shown
above, and  then $\Del\der \vv:\mult{\prodt{\A}{\D}}$, by rule $(\many)$. Thus $\Gam +\Del \der \s'[\pair{\p}{\q}/\vv]:\sigma$ by
applying rule $(\trsub)$.
\item (only if) Let $\s$ be a typable term, \ie\ $\Pi \dem \Gam \der
  \s:\sig$.  Consider a reduction strategy ${\cal ST}$ that always
  chooses a {\it typed} redex occurrence. By
  Lemma~\ref{lem:red:exp}(\ref{lem:reduction}) and Lemma~\ref{lem:inf}
  the strategy ${\cal ST}$ always terminates. Let $\s'$ be a
  normal-form of $\s$ for the strategy ${\cal ST}$, \ie\ $\s$ reduces
  to $\s'$ using ${\cal ST}$, and $\s'$ has no typed redex occurrence.  We know that
  $\Pi' \dem \Gam \der \s': \sig$ by
  Lemma~\ref{lem:red:exp}(\ref{lem:reduction}).
  We now proceed by
  induction on $\Pi'$, by taking into account the notion of typed
  occurrence of $\Pi'$.

  If $\Pi'$ ends with $(\ax)$, then its subject is $\x$, which is
  canonical.  If $\Pi'$ ends with $(\introarrow)$ with subject $\lambda
  \p.\uu$ and premise $\Pi''$ with subject $\uu$, then
  $\uu$ has no typed
    redex occurrences, so it is canonical by the \ih\ 
  We conclude that $\lambda \p.\uu$ is canonical too, by
  definition of canonical form.  If $\Pi$ ends with $(\app)$ with
  subject $\s\uu$ and premises $\Pi_1$ and $\Pi_2$ having
  subjects $\s$ and $\uu$ respectively, then $\s$, which is
    also typable, has no typed redex occurrences, so that it is
    canonical by the \ih\  Moreover $\square \in \toc{\Pi}$, so
  $\s\uu$ cannot be a redex. This implies that $\s$ cannot be an abstraction, so it is a ${\cal K }$ canonical form,  and consequently $\s\uu$ is a ${\cal K}$ canonical form too.  Suppose $\Pi$ ends with
  $(\trsub)$ with subject $\s[\p/\uu]$ and premises $\Pi_1$ and
  $\Pi_2$ with subjects $\s$ and $\uu$
  respectively. The term $\s$, which is typable,  has no typed redex occurrences, so
it is canonical by the \ih\ 
Moreover, $\p$ cannot be a variable, otherwise the term would have a
typed $\rdeux$-redex occurrence, so, by Corollary~\ref{cor:gen-pair},
$\Gam|_{\p}\pder \p:\mult{\prodt{\A}{\D}}$, for some $\A,\D$, where $\Gam$ is the
typing environment of $\s$. The term   $\uu$ is
also typed. Since $\uu$ cannot be neither an abstraction
(rule $\rquatre$) nor a pair (rule $\rtrois$), it is
necessarily a ${\cal K}$ canonical form, and consequently $\s[\p/\uu]$
is canonical. Finally, if $\Pi'$ ends with $(\trpair)$,
then its subject is a pair, which is a
  canonical form. This concludes the proof.  \qedhere
\end{itemize}
\end{proof}

\section{Inhabitation for System $\Pu$}
\label{s:inhabitation}

Given $\mu$, a type or a multiset type, the inhabitation problem consists
in finding a closed term $\s$ such that $\der \s: \mu$ is derivable.
These notions will naturally be generalized later to non-closed terms.
Since system $\Pu$ characterizes canonicity, it is natural to look for
inhabitants in {\it canonical form}. The next Lemma proves that the
problem can be simplified, namely that it is sufficient to look for
inhabitants in {\it pure canonical form}, \ie\ without nested
substitution (we postpone the proof of this lemma to
Section~\ref{sub:pure}).

\begin{restatable}{lemma}{lempure} \label{lem:pure} Let $\s$ be a canonical form.  If
  $\Pi \dem \Gam \der \s: \mu$ is derivable, then there is some type
  derivation $\Pi'$ and some {\it pure} canonical form $\s'$ such that
  $\Pi' \dem \Gam \der \s':\mu$ is derivable.
  \end{restatable}

  We already noticed that the system $\Pu$ allows to assign
  the multiset $\emul$ to terms through the rule $(\many)$.
  As a consequence, a typed term may contain untyped subterms. 
In order to  identify 
inhabitants in such cases we  introduce a term constant $\Omega$ to denote a generic
untyped subterm. Accordingly, the type system $\Pu$ is extended to
the new grammar of terms possibly containing $\Omega$, which can
only be typed using  a particular case of the $(\many)$ rule:
\[ \infer{ }{\der \Omega: \emul}\ (\many)
\]

So the inhabitation algorithm should produce 
\deft{approximate normal forms} (denoted $\ap, \bp, \cp$), also written $\anf$, defined as follows: 
\[
\begin{array}{lll}
     {\ap},{\bp},{\cp} & ::= & \Omega \mid {\cal N} \\ 
     {\cal N}          & ::= & \lambda \p. {\cal N}  \mid  
                               \pair{\ap}{\bp}  \mid 
                               {\cal L} \mid   {\cal N} [\pair{\p}{\q}/{\cal L}]\\
  {\cal L}          & ::= &  \x  \mid {\cal L} \ap   \\
   \end{array}  
\]
The grammar defining $\anf s$ is similar to
   that of pure canonical forms, starting, besides variables, also
   from $\Omega$. 
 The notion of typed occurrences in 
the new extended system is  straightforward.    
  Moreover,  an $\anf$
     does not contain any redex, differently from canonical
     forms. Roughly speaking, an $\anf$ can be seen as a representation
     of an infinite set of pure canonical forms, obtained by replacing
     each occurrence of $\Omega$ by any term.
     \begin{example}
The term  $\lambda \pair{\x}{\y}.   (\x    (\id    \id))
[\pair{\z_1}{\z_2}/\y\id]$ is (pure) canonical but not an $\anf$, 
while  $\lambda  \pair{\x}{\y}. (\x \Omega) [\pair{\z_1}{\z_2}/\y\id]$ is
an    $\anf$.
\end{example}

$\Anf s$ 
are ordered by  the smallest contextual  
order $\leq$ such that $\Omega \leq \ap$, for any  $\ap$.
We also write $\ap \leq \s$ when the term $\s$ is obtained from 
 $\ap$ by
 replacing each occurrence of $\Omega$ by a term of $\Lp$.
 Thus for example $\x \Omega \Omega \leq \x (\id \dup) (\dup \dup)$ is obtained by replacing 
the first (resp. second) occurrence of $\Omega$ by $\id \dup$ (resp. $\dup \dup$). 

Let $\Ap(\s)=\{\ap \ |\ \exists \uu\ \s \Rewn{} \uu \mbox{ and
}\ap\leq \uu\}$ be the set of \deft{approximants}  of the term $\s$, and
let $\bigvee$ denote the least upper bound with respect to $\leq$.  We write
$\uparrow_{\iI}\ap_i$ to denote the fact that $\bigvee\{\ap_i\}_{\iI}$
does exist. Note that $I =\es$ implies $\bigvee\{\ap_i\}_{\iI} = \Omega$.  It is easy to check that, for every $\s$ and
$\ap_1,\ldots \ap_n \in \Ap(\s)$,
$\uparrow_{i\in\{1,\ldots,n\}}\ap_i$. 
An $\anf$  $\ap$
is a \deft{head subterm} of $\bp$ if either $\bp = \ap$ or $\bp=\cp
\cp'$ and $\ap$ is a head subterm of $\cp$.  It is easy to check that, if
$\Gamma\der \ap:\sigma$ and $\ap\leq\bp$ (resp.  $\ap\leq\s$) then
$\Gamma\der \bp:\sigma$ (resp.  $\Gamma\der \s:\sigma$). 

 Given
$\Pi\dem\Gam   \der  \s:\mu$,   where  $\s$   is  in   $\Pi$-nf
(\cf\ Section~\ref{s:type-system}),  the
\deft{minimal approximant of $\Pi$}, written $\Ap(\Pi)$, is defined by
induction on $\Pi$ as follows:
\begin{itemize}
\item $\Ap(\Gam \der \x:\rho)= \x$.
\item If $\Pi\dem\Gam \der \lambda \p. \s: \A \arrow \rho$ follows  from 
  $\Pi' \dem\Gam' \der \s:\rho$, then $\Ap(\Pi)= \lambda \p. \Ap(\Pi')$,  $\s$ being in  $\Pi'$-nf. 
\item If $\Pi\dem \Gam \der \pair{\s}{\uu}:\prodt{\A}{\D}$ follows
  from $\Pi_1 \dem \Gam_1 \der \s: \A$ and  $\Pi_2 \dem \Gam \der \uu: \D$, then
  $\Ap(\Pi)= \pair{\Ap(\Pi_1)}{\Ap(\Pi_2)}$. 
\item If $\Pi\dem\Gam=\Gam'+\Delta\der \s\uu:\rho$ follows  from $\Pi_1\dem\Gam'\der \s:  \A \to \rho$
  and $\Pi_2\dem\Delta\der\uu:\A$,
  then $\Ap(\Pi)=\Ap(\Pi_1)(\Ap(\Pi_2))$.
\item If $\Pi\dem \Gam = \Gam' +\Del\der \s[\p/\uu]: \tau$
  follows  from $\Pi'\dem \Gam'' \der \s:\tau$ and 
  $\Psi \dem \Del \der \uu:\A$, then 
  $\Ap(\Pi)=  \Ap(\Pi')[\p/\Ap(\Psi)]$.
\item If $\Pi \dem +_{\iI}\Gam_i \der \s: \mult{\sig_i}_{\iI}$ follows from $(\Pi_i \dem \Gam_i\der \s: 
  \sig_i)_{\iI}$, then $\Ap(\Pi)=\bigvee_\iI\Ap(\Pi_i )$.
\end{itemize}

\begin{example}
  Consider the following derivation $\Pi$ (remember that $\oprod$ is an abbreviation for $\prodt{\emul}{\emul}$), built upon the subderivations $\Pi_1$
  and $\Pi_2$ below:
\[
  \infer{
    \infer{
      \Pi_1 \sep
      \infer{}{\pder \pair{\z_1}{\z_2}:\mult{\oprod}} 
      \sep 
      \Pi_2}
    {\x:\mult{\emul \to \oprod} , \y:\mult{\emul \to \oprod} \der \y(\dup \dup)
      [\pair{\z_1}{\z_2}/\x\id]:\oprod}}{\der \lambda \x\y.\y (\dup \dup)
    [\pair{\z_1}{\z_2}/\x\id]: \mult{\emul \to \oprod} \to \mult{\emul \to
      \oprod} \to \oprod}
\]
\vspace{.5cm}
\[
  \Pi_1 = \infer{
      \infer{}{\y:\mult{\emul \to \oprod}\der \y:\emul \to \oprod} \sep
      \infer{}
            {\es \der \dup \dup: \emul} }
                   {\y:\mult{\emul \to \oprod} \der \y  (\dup \dup):\oprod}
                   \hspace{1em}
                   \Pi_2 = \infer{\infer{
                       \infer{}{\x:\mult{\emul \to \oprod}\der
                         \x:\emul \to \oprod} \sep \infer{}
                       {\es \der \id: \emul}}
                     {\x:\mult{\emul \to \oprod}\der \x\id: \oprod}}
                    {\x:\mult{\emul \to \oprod}\der \x\id:\mult{ \oprod}}
\]
The minimal approximant of $\Pi$ is $\lambda \x\y.\y\Omega[\pair{\z_1}{\z_2}/\x\Omega]$.
\end{example}
A simple induction on  $\meas(\Pi)$ allows to show the following:

\begin{lemma}
\label{Prop:approx}
If $\Pi \dem \Gam \der \s:\mu$ and $\s$ is in $\Pi$-nf,  then    
$\Pi \dem \Gam \der \Ap(\Pi):\mu$.
\end{lemma}

\subsection{From Canonical Forms to Pure Canonical Forms}
\label{sub:pure}

In this section we prove that, when a giving typing is inhabited, then
it is necessarily inhabited by a pure canonical form. This property
turns out to be essential to prove  the completeness property of our
algorithm (Theorem~\ref{t:soundness-completeness}),  since the algorithm
only builds pure canonical forms.

\lempure*

  \begin{proof}
  By induction on $\Pi \dem \Gam \der \s: \mu$.
  The only interesting case is when $\s = \s_0[\p/\s_1]$, where  $\p$ is some pair pattern.
  By construction of $\Pi$, there is a type derivation of the following form:
 \[\infer{\Gam \der \s_0 : \tau \sep 
            \Gam|_{\p} \pder \p: \mult{\prodt{\A}{\D}} \sep   
          \infer{\Del \der \s_1: \prodt{\A}{\D}}{\Del \der \s_1:\mult{\prodt{\A}{\D}}}}
        { \Gam \sm \Gam|_{\p} + \Del\der \s_0[\p/\s_1]: \tau}
\]
where the shape of the type for $\p$ comes from Corollary \ref{cor:gen-pair}.
         By the \ih\ there are pure canonical terms $\s'_0, \s'_1$ such that
         $ \Gam \der \s'_0 : \tau$ and $ \Del\der \s'_1:\prodt{\A}{\D}$.
         If  $\s'_0[\p/\s'_1]$ is a pure canonical term, then we conclude
         with a derivation of $\Gam \sm \Gam|_{\p} + \Del \der \s'_0[\p/\s'_1]: \tau $.
         If $\s'_0[\p/\s'_1]$ is not a pure canonical term,
         then necessarily $\s'_1 = \uu[\q / \vv]$, $\q$ being a pair and $\uu, \vv$ being pure canonical terms.
        Then there is necessarily  a derivation of the following form:
        \[
    \infer{\Del'\der \uu : \prodt{\A}{\D} \sep 
            \Del'|_{\q} \pder \q: \mult{\prodt{\C}{\E}} \sep   
         \infer{\Gam' \der \vv: \prodt{\C}{\E}}{\Gam' \der \vv:\mult{\prodt{\C}{\E}}} }
        { \Del= \Del' \sm \Del'|_{\q} + \Gam' \der \uu[\q / \vv]:  \prodt{\A}{\D} }
  \]

We can then  build the following derivation:
\[
\infer{\Gam \der \s'_0 : \tau \sep 
            \Gam|_{\p} \pder \p: \mult{\prodt{\A}{\D}}  \sep   
          \infer{\Del' \der \uu: \prodt{\A}{\D}}{\Del '\der \uu:\mult{\prodt{\A}{\D}}} }
        { \Gam \sm \Gam|_{\p} +\Del'\der \s'_0[\p/\uu]: \tau }
        \]
 Note that $(\Gam \sm \Gam|_{\p} +\Del')|_\q = \Del'|_\q$, since we
 can always choose $\fv{\q} \cap \Gam = \emptyset$, by $\alpha$ conversion.
 So the following derivation can be built:
\[
\infer{ \Gam \sm \Gam|_{\p} +\Del' \der \s'_0[\p/\uu] : \tau \sep 
            \Del'|_\q \pder \q: \mult{\prodt{\C}{\E}}  \sep   
          \infer{\Gam' \der \vv: \prodt{\C}{\E}}{\Gam'\der \vv:\mult{\prodt{\C}{\E}}} }
        { ( \Gam \sm \Gam|_{\p} +\Del' ) \sm \Del'|_\q + \Gam'
          \der \s'_0[\p/\uu][\q/\vv]: \tau }
        \]
Since $\Gam \sm \Gam|_{\p} \sm \Del'|_\q = \Gam \sm \Gam|_{\p}$ and $\Del' \sm \Del'|_\q + \Gam' = \Del$, the proof is given.  
\end{proof}

\subsection{The Inhabitation Algorithm}

We now show a sound and complete algorithm to solve the inhabitation
problem for System $\Pu$.  
The  algorithm is presented in
Figure~\ref{fig:inhabitation-algorithm}. As usual, in order to solve
the problem for closed terms, it is necessary to extend 
the algorithm  to open
ones, so, given an environment $\Gam$ and a type $\sigma$,
the algorithm builds the set $\K(\Gam, \sig)$ containing {\it all} the
$\anf s$ $\ap$ such that there exists a derivation
$\Pi\dem\Gam \der \ap:\sig$, with $\ap=\Ap(\Pi)$, then stops\footnote{
It is worth noticing that,  given $\Gam$ and $\sig$, the set of $\anf s$   $\ap$ such that there exists a derivation $\Pi\dem\Gam \der \ap:\sig$
is possibly infinite. However, the subset of those verifying  $\ap=\Ap(\Pi)$ is finite.This is proved in Corollary~\ref{cor:finite}.}. Thus,
our algorithm is not an extension of the classical inhabitation
algorithm for simple types~\cite{BenYellesPhd,hindley08}. In
particular, when restricted to simple types, it constructs all the
$\anf s$  inhabiting a given type, while the original algorithm
reconstructs just the {\it long $\eta$-normal forms}.
The algorithm  uses three  auxiliary predicates, namely 
\begin{itemize}

\item
$\Pa{}{\V}(\A)$, where $\V$ is a finite set of variables, contains  the pairs $(\Gam,\p)$ 
such that  (i) $\Gam\pder\p:\A$, and (ii)
$\p$ does
not contain any variable in $\V$.
\item
$\KI(\Gam,\mult{\sig_i}_{\iI})$, contains all the $\anf s$ $\ap=\bu_{\iI} \ap_i$ such that  $\Gam =  +_{\iI} \Gam_i$,
$\ap_i\in\K(\Gam_i,\sigma_i)$ for all $i\in I$, and $\uparrow_{\iI}\ap_i$.
\item
$\LK{\bp}{\Del}(\Gam, \sig) \donne \tau$  contains all the $\anf s$ $\ap$ such that 
$\bp$ is a head subterm of $\ap$, and such that 
if $\bp \in \K(\Delta,\sigma)$ then 
$\ap \in \K(\Gamma+\Delta,\tau)$. As a particular case, notice that
$\bp \in \LK{\bp}{\Del}(\es, \sig) \donne \sig$, for all $\bp \in {\cal L}$,
environment $\Del$ and type $\sig$.
\end{itemize}
\begin{figure}[h!]
\begin{framed}
\begin{center}
$ \begin{array}{c}
\infer{\x \notin \V}
      { (\x:\A,   \x) \in \Pa{}{\V}(\A) }\ (\Varp) \\\\
\infer{(\Gam ,   \p) \in \Pa{}{\V}(\A) \sep  
 (\Delta ,   \q) \in \Pa{}{\V}(\D) \sep  
 \p \# \q}
{( \Gam + \Delta,  \pair{\p}{\q}  ) \in   \Pa{}{\V}(\mult{\prodt{\A}{\D}})}\ (\Pairp) \\\\
\infer
   {\ap \in \K(\Gam +   \Delta ,   \tau)  \sep  
   (\Delta,  \p)  \in  \Pa{}{\dom{\Gam}}(\A)}
{ \lambda \p. \ap \in \K(\Gam  ,  \A \arrow \tau)}\ (\Abs) \sep 
\infer{\ap \in \KI(\Gam, \A) \sep 
         \bp \in \KI(\Del, \D)
      }{ \pair{\ap}{\bp} \in \K(\Gam+\Delta,\prodt{\A}{\D})}\ (\Prod) \\\\
\infer{ (\ap_i \in \K( \Gam_i,  \sig_i))_{\iI} \sep \uparrow_{\iI}\ap_i }
         {\bigvee_{\iI} \ap_i \in  \KI(+_{ \iI} \Gam_i,  \mult{\sig_i}_{\iI})}\ (\Unionk) 
   \\ \\

\infer{ \ap \in \LK{\x}{\x: \mult{\sig}}( \Gam ,  \sig) \donne   \tau }
         { \ap  \in \K( \Gam +  \x: \mult{\sig}, \tau)}\ (\Head)
\quadd
\infer {\sig = \tau}{\ap \in \LK{\ap}{\Del}(\es, \sig) \donne  \tau}\ (\Final) \\\\
\infer{   \bp \in \KI(\Gam, \A) \sep 
            \ap \in \LK{\cp \bp}{\Del + \Gam }(\Lam, \sig) \donne \tau  }
            {  \ap\in \LK{\cp}{\Del}( \Gam +\Lam, \A \to \sig) \donne  \tau   }\ (\Prefix)\\\\
\infer{ \cp \in \LK{\x}{\x:\mult{\sig}}(\Gam, \sig) \donne \fin{\sig}(*) \minisep 
                    (\Del, \pair{\p}{\q}) \in \Pa{}{\dom{\Gam + \Lam  + ( \x: \mult{\sig}) }}(\mult{\fin{\sig}}) \minisep 
                    \bp \in \K( \Del +  \Lam, \tau)  }
         {\bp[\pair{\p}{\q}/\cp] \in \K( \Gam+\Lam  +  \x: \mult{\sig}, \tau)} (\Subs) \\\\
\end{array}$ 
(*) where the operator \fin{} on types is defined as follows:

$\begin{array}{lll@{\hspace{.8cm}}lll}
\fin{\alpha} &  := & \alpha \\ 
\fin{\PT} &  := & \PT \\
\fin{\A \to \tau} & := & \fin{\tau} & \\
\end{array}$
\end{center}
\caption{The inhabitation algorithm}
                    
\label{fig:inhabitation-algorithm}
\end{framed}
\end{figure}
Note that a special case of rule $(\Unionk)$
with $I = \es$
is $\infer{\phantom{\Omega \in}  }{\Omega \in \KI(\emptyset, \emul)}$.
Note also that the algorithm has different kinds of non-deterministic
behaviours, \ie\ different choices of rules can produce different
results.  Indeed, given an input $(\Gam, \A \to\sig)$, the algorithm
may apply a rule like $(\Abs)$ in order to decrease the type $\sig$,
or a rule like $(\Head)$ in order to decrease the environment $\Gam$.
Moreover, every rule $(R)$ which is based on some decomposition of the
environment and/or the type, like $(\Subs)$, admits different
applications.  In what follows we illustrate the non-deterministic
behaviour of the algorithm.  For that, we represent a {\bf run of the
  algorithm} as a tree whose nodes are labeled with the name of the
rule being applied.

\begin{example}\label{Exe1}
We consider different inputs of the form $(\es, \sig)$, for different
types $\sig$. For every such input, we give an output and the corresponding run.
\begin{enumerate}
\item\label{primo}
$\sigma =\mult{\mult{\alpha} \to \alpha}\arrow \mult{\alpha} \to \alpha$.
\begin{enumerate}
\item output: $\lambda \x\y.\x\y$, run: \\
  $\Abs(\Abs(\Head(\Prefix(\Unionk(\Head(\Final)), \Final)),\Varp),\Varp)$.
\item output: $\lambda \x.\x$, run: \\
  $\Abs(\Head(\Final),\Varp)$.
\end{enumerate}
\item $\sigma =\mult{\emul \to \alpha}\arrow \alpha$. output:  $\lambda\x.\x\Omega$, run: $\Abs(\Head(\Prefix(\Unionk, \Final)),\Varp)$.
\item $\sigma =\mult{\mult{\oprod} \to \oprod,\oprod}\arrow \oprod$.
\begin{enumerate}
\item output: $\lambda \x.\x\x$, run:
$\Abs(\Head(\Prefix(\Unionk(\Head(\Final)), \Final)),\Varp)$. 
\item Explicit substitutions may be used to consume some, or all, the resources in $\mult{\mult{\oprod} \to \oprod,\oprod}$.
output: $\lambda \x.\x[\pair{\y}{\z}/\x\pair{\Omega}{\Omega}]$, run:\\
{\small 
$\Abs(\Subs(\Prefix(\Unionk(\Prod),\Final),\Pairp(\Varp,\Varp),\Head(\Final)), \Varp)$.
}
\item There are four additional runs, producing the following outputs:

$\lambda \x.\x\pair{\Omega}{\Omega}[\pair{\y}{\z}/\x]$, 

$\lambda \x.\pair{\Omega}{\Omega}[\pair{\y}{\z}/\x\x]$, 

$\lambda \x.\pair{\Omega}{\Omega}[\pair{\y}{\z}/\x][\pair{\w}{\tt{s}}/\x\pair{\Omega}{\Omega}]$,
 
$\lambda \x.\pair{\Omega}{\Omega}[\pair{\y}{\z}/\x\pair{\Omega}{\Omega}] [\pair {\w}{\tt{s}}   /\x]$.
\end{enumerate}
\end{enumerate}
\end{example}

Along the recursive calls of the inhabitation algorithm, 
the parameters (type and/or environment)
decrease  strictly, for a suitable notion of measure,  
so that every run is finite:


\begin{restatable}{theorem}{thtermination}
\label{th:termination}
The inhabitation algorithm terminates.
\end{restatable} 
\begin{proof}
See Subsection~\ref{sub:termination}.
\end{proof}


We now prove soundness and completeness of our inhabitation algorithm.

\begin{lemma}
\label{Lem:main}
$\ap\in\K(\Gamma,\sigma)$ $\Leftrightarrow$
$\exists \Pi \dem \Gam \der \ap:\sigma$ such that $\ap=\Ap(\Pi)$.
\end{lemma}

\begin{proof} 
  The ``only if'' part is proved by induction on the rules in
  Figure~\ref{fig:inhabitation-algorithm}, and the ``if'' part is
  proved by induction on the definition of $\Ap(\Pi)$ (see
    Section~\ref{sub:termination} for full details).  In both parts, additional
  statements concerning the predicates of the inhabitation algorithm
  other than $\K$ are required, in order to strengthen the inductive
  hypothesis.
\end{proof}

\begin{theorem} [Soundness and Completeness]
\label{t:soundness-completeness} \mbox{}
\begin{enumerate}
\item 
\label{t:soundness}
If $\ap \in \K(\Gam, \sig)$  then, for all $\s$ such that $\ap\leq\s$,  $\Gam  \der \s:\sig$.

\item
 \label{t:completeness}
If $\Pi\dem\Gam \der \s:\sig$ then there exists  $\Pi'\dem\Gam \der \s':\sig$ such that  $\s'$ is in $\Pi'$-nf, 
and  $\Ap(\Pi') \in \K(\Gam,\sig)$.

\end{enumerate}
\end{theorem}
\begin{proof}
Soundness follows from Lemma~\ref{Lem:main} ($\Rightarrow$) and the
fact that $\Gamma\der \ap:\sigma$ and $\ap\leq\s$ imply $\Gamma\der
\s:\sigma$. 
Concerning completeness: Theorem \ref{l:characterization-canonical}
and Lemma~\ref{lem:red:exp}(\ref{lem:reduction})
ensures that
$\s$ has a canonical form $\s_0$
such that $\Pi_0\dem\Gam \der \s_0:\sig$. Then,  
  Lemma~\ref{lem:pure} guarantees the existence of a pure canonical form $\s'$ such that 
  $\Pi'\dem\Gam \der \s':\sig$ and $\s'$ is in $\Pi'$-nf. Then 
  Lemma~\ref{Prop:approx} and Lemma~\ref{Lem:main} ($\Leftarrow$) allow us to conclude.
\end{proof}

\subsection{Properties of the Inhabitation Algorithm}
\label{sub:termination}

We prove several properties of the inhabitation algorithm,
  namely, termination, soundness and completeness.

\paragraph{Termination} Being the inhabitation algorithm non
deterministic, proving its termination means to prove both that a
single run terminates and that every input generates a finite number
of runs. We will prove these two properties separately.

First, let us define the following \deft{measure} on types and environments:
\[
   \begin{array}{lll}
   \ccount{\alpha} & := & 1 \\
  \ccount{\mult{\sig_i}_{\iI}} & := & 1 +_{\iI} \ccount{\sig_i}  \\
  \ccount{\prodt{\A}{\D}} & := & \ccount{\A} + \ccount{\D} +1\\
  \ccount{\A \to \sig } & := & \ccount{\A} + \ccount{\sig} + 1\\
   \ccount{\Gam} & := & \sum_{\x \in \dom{\Gam}} \ccount{\Gam(\x)} \\
    \end{array}
\]

The measures are extended to predicates in the following way:
\[
   \begin{array}{lll}
\ccount{\K(\Gam, \sig)} & : =& \ccount{\Gam} + \ccount{\sig}\\
\ccount{\LK{\ap}{\Del}(\Gam, \sig) \donne  \tau} & := & \ccount{\Gam} + \ccount{\sig}\\
  \ccount{\KI(\Gam, \A)} & := & \ccount{\Gam} + \ccount{\A} \\
  \ccount{\Pa{}{\V}(\A) } & := & \ccount{\A} \\
    \end{array}
\]

Notice that $\ccount{\Gam} \leq \ccount{\Gam + \Del}$, for
  any $\Del$. Also,  $\ccount{\Gam + \Del} \leq \ccount{\Gam} +
  \ccount{\Del}$, thus \eg\ $ \ccount{\x:\mult{\alpha} + \x:\mult{\alpha} } =
  \ccount{\x:\mult{\alpha, \alpha}} = 3 \leq \ccount{\x:\mult{\alpha}} + \ccount{\x:\mult{\alpha}} = 4$. Notice also that
  $\ccount{\Delta_1} < \ccount{\Del_2}$ does not imply
  $\ccount{\Del_1 + \Lam} < \ccount{\Del_2+\Lam}$, \eg\
 when  $\Del_1 = \x:\mult{\alpha}$, $\Del_2 = \y:\mult{\alpha,\alpha}$
 and $\Lam = \y: \mult{\alpha}$. However, as a particular useful case, 
 if $ \ccount{\Delta}+1 <  \ccount{\x:\A}$,
 then $\ccount{\Del + \Lam} < \ccount{\x:\A+\Lam}$.
 Indeed, if $\x \notin \dom{\Lam}$,
 then  $\ccount{\Del + \Lam}\leq   \ccount{\Del} + \ccount{\Lam}
 < \ccount{\x:\A}+ \ccount{\Lam} = \ccount{\x:\A+\Lam}$;
  otherwise,  $\ccount{\Del + \Lam}\leq   \ccount{\Del} + \ccount{\Lam}
  < 
  (\ccount{\x:\A} -1)  + \ccount{\Lam}  = \ccount{\x:\A+\Lam}$.

The following property follows directly.

\begin{property}\label{prop:ccount}
  Let
 $(\Gam,  \p)  \in  \Pa{}{\dom{\V}}(\A)$.
      Then   $\ccount{\Gam}\leq \ccount{\A}$.
      Moreover,  $\p=\pair{\p_{1}}{ \p_{2}}$ implies $\ccount{\Gam} +1 < \ccount{\A}$.
      \end{property}

We can now prove:

\begin{lemma}\label{lem:run-termination}
Every run of the algorithm terminates.
\end{lemma}
\begin{proof}
  We associate a tree $\tree$ to each call of the algorithm, where the
  nodes are labeled with elements in the set $\{\K(\_
  ,\_),\KI(\_,\_),\LK{\_}{\_}(\_ ,\_) \donne \_ , \Pa{}{\_}(\_) \}$.
  A node $n'$ is a son of $n$ iff there exists some instance of a rule
  having $n$ as conclusion and $n'$ as premise. Thus, a run of the
  algorithm is encoded in the tree $\tree$, which turns to be
  finitely branching.  We now prove that the measure $\ccount{\_}$
  strictly decreases along all the branches of $\tree$, so that every
  branch has finite depth. We proceed by induction on the rules of the
  algorithm. The only interesting cases are rules $(\Abs)$ and
  $(\Subs)$.

  \begin{itemize}
    \item Consider rule $(\Abs)$, with conclusion $\lambda \p. \ap \in
      \K(\Gam , \A \arrow \tau)$ and premises  $\ap \in \K(\Gam + \Delta , \tau)
      $ and $(\Delta, \p) \in
      \Pa{}{\dom{\Gam}}(\A)$. By Property~\ref{prop:ccount}, $\ccount{\Delta} \leq
      \ccount{\A}$, so that 
      $\ccount{\K(\Gam + \Delta , \tau) } 
  \leq 
  \ccount{\Gam}
  +\ccount{\Delta}+\ccount{\tau} \leq \ccount{\Gam}
  +\ccount{\A}+\ccount{\tau}   <    
  \ccount{\K(\Gamma, \A \arrow \tau)}$
  and $\ccount{\Pa{}{\dom{\Gam}}(\A)} =
          \ccount{\A}  <  \ccount{\K(\Gamma, \A \arrow \tau)}$.
 \item Consider  rule $ (\Subs)$,
  with conclusion $\bp[\pair{\p}{\q}/\cp] \in \K( \Gam+\Lam +
  \x: \mult{\sig}, \tau)$ and premises $\cp \in
  \LK{\x}{\x:\mult{\sig}}(\Gam, \sig) \donne \fin{\sig}$, $(\Del,
  \pair{\p}{\q}) \in \Pa{}{\dom{\Gam + \Lam + ( \x: \mult{\sig})
  }}(\mult{\fin{\sig}})$ and $\bp \in \K( \Del + \Lam, \tau)$.
  Clearly $\ccount{ \LK{\x}{\x:\mult{\sig}}(\Gam, \sig) \donne
    \fin{\sig}} = \ccount{\Gam} + \ccount{\sig} <
   \ccount{\Gam} + \ccount{\Lam} + \ccount{\mult{\sig}} + \ccount{\tau}  = \ccount{\K( \Gam+\Lam + \x: \mult{\sig}, \tau)}$.
  Also, $\ccount{\Pa{}{\dom{\Gam + \Lam + ( \x: \mult{\sig})
      }}(\mult{\fin{\sig}})} = \ccount{\mult{\fin{\sig}}} \leq
     \ccount{\mult{\sig}}
    < \ccount{\K(
      \Gam+\Lam + \x: \mult{\sig}, \tau)}$.
    Finally, by Property~\ref{prop:ccount},
  $\ccount{\Del} +1 < \ccount{
    \mult{\fin{\sig} }} \leq \ccount{\mult{\sig}} = \ccount{\x:\mult{\sig}}$.
  So $\ccount{\K( \Del
    + \Lam, \tau)} = \ccount{\Del + \Lam} + \ccount {\tau}<
  \ccount{\x:\mult{\sig} + \Lam} + \ccount {\tau} \leq  \ccount{\Gam+ \x:\mult{\sig} + \Lam} + \ccount {\tau} =\ccount{\K(
    \Gam+\Lam + \x: \mult{\sig}, \tau) } $.
    \end{itemize}
So every branch has finite depth. Hence, $\tree$ is finite by
K\"onig's Lemma, \ie\ the algorithm terminates. 
\end{proof}

In order to complete the proof of termination we need to show that the 
number of different run of the algorithm on any  given input  is finite.

Let $\Pi \dem \Gam \der \ap:\sig$, where, by
  $\alpha$-conversion, we assume that $\fv{\ap} \cap \bv{\ap} =
  \emptyset$. We write $||\ap||^\Pi_\x$ (resp $|\ap|^\Pi_\x$) to
  denote 
the number of free (resp.  bound) occurrences of
  $\x$ in $\ap$ which are typed in $\Pi$.
The following property holds\footnote{Tighter upper bounds than those provided below may be found, but this is inessential here.}.
\begin{property}\label{prop:times}
Let $\ap$ be an approximate normal form.  Let $\Pi \dem \Gam \der
\ap:\sig$. Then, for every variable $\x$ occurring in $\ap$
we have $||\ap||^\Pi_\x\leq \ccount{\Gam(\x)}$ and $|\ap|^\Pi_\x \leq \ccount{\Gam} + \ccount{\sig}$.
\end{property}

\begin{proof} 
  If $\x \in \fv{\ap}$, then $\x:\multiset{\sig_i}_{\iI} \in \Gamma$,
  for some non-empty set $I$, and since every axiom corresponds to a free
  occurrence of $\x$ in  $\ap$ which is typed in $\Pi$, then
 the number of such occurrences is exactly the
  cardinality of $I$, which is trivially smaller than $\ccount{\Gam(\x)}$.

Let $\x \in \bv{\ap}$. The proof is by induction on $\Pi$.
  \begin{itemize}
    \item Let the last rule of $\Pi$ be $(\introarrow)$, with conclusion
  $\Pi\dem \Gam'\sm \Gam'|_{\p} \der \lambda \p.\bp:\A \arrow \tau$ and
      premises $\Pi'\dem\Gam' \der \bp: \tau$ and $\Gam'|_{\p}\pder \p:\A$, where  $\Gam = \Gam'\sm \Gam'|_{\p}$
      and $\sig = \A \arrow \tau$.
      Since  $\x$ is bound in $\lambda\p.\bp$, then either
      $\x$  is bound in $\bp$ or $\x$ occurs in the pattern
  $\p$. If $\x$ is bound in $\bp$, then the
        proof follows by induction. Otherwise,  by Lemma~\ref{l:patterns},
      $\Gam'|_{\x}\pder \x:\A^{\p}_{\x}$,
  \ie\ $\Gam'(\x)=\A^{\p}_{\x}$.  Then  the number of free occurrences
  of $\x$ in $\bp$ typed in $\Pi'$ is
  $||\bp||^{\Pi'}_\x\leq\ccount{\A^{\p}_{\x}}$, so the number of its bound
  occurrences in $\lambda \p.\bp$ typed in $\Pi$ is the same.  Since
  $\ccount{\A^{\p}_{\x}} \leq \ccount{\A}$ and $\ccount{\A} \leq \ccount{\A\arrow \tau} + \ccount{\Gam}$,
 then we are done. 

\item Let the last rule of $\Pi$ be $(\trsub)$, with conclusion
  $\Gam'\sm \Gam'|_{\pair{\p}{\q}}+\Delta \der
  \bp[\pair{\p}{\q}/\cp]:\sig$ and premises $\Pi_{\bp}\dem
  \Gam'\der\bp: \sig$, $\Gam'|_{\pair{\p}{\q}}\der
  \pair{\p}{\q}: \mult{\prodt{\A}{\D}}$, and $\Pi_{ \cp}\dem \Delta \der
  \cp: \mult{\prodt{\A}{\D}}$, where $\ap= \bp[\pair{\p}{\q}/\cp]$ and $\Gam = \Gam'\sm
    \Gam'|_{\pair{\p}{\q}}+\Delta$.  Since $\x$ is bound in
  $\bp[\pair{\p}{\q}/\cp]$, then either $\x$ is bound
  in $\bp$ or $\cp$, or $\x$ occurs in $\pair{\p}{\q}$
  and so $\x$ is free in $\bp$.  If $\x$ is bound in
  $\bp$ or $\cp$, then the proof follows by induction.  Let $\x$ occur
  in $\pair{\p}{\q}$.  Since $\x$ occurs free in $\bp$, by the \ih,
  $||\bp||^{\Pi_\bp}_\x\leq\ccount{\Gamma'(\x)}$.  Notice that,
  by definition of approximant, $\cp$ must be an ${\cal L}$
  approximate normal form, so that let
  $\cp=\y\ap_{1}...\ap_{n}$ ($n\geq 0$). The derivation $\Pi_{ \cp}$
  has necessarily been obtained by applying rule $(\many)$ to a
  derivation $\Pi'_{ \cp}\dem \Delta \der \cp: \prodt{\A}{\D}$, so
  $\Delta(\y)$ must be $\C+ \mult{\A_{1}\arrow \ldots \A_{n}\arrow
    \prodt{\A}{\D}}$, for some $\C, \A_{1}, \ldots ,\A_{n}$ such that
  $\Delta_{i}\der \ap_{i}:\A_{i}$, and $\Del = +_{\iI}\Del_{i} + \y :
  \mult{\A_{1}\arrow \ldots \A_{n}\arrow \prodt{\A}{\D}} $,
  where $(+_{\iI}\Del_{i})(y)=\C$. So
  $\ccount{\Del}\geq \ccount{\mult{\prodt{\A}{\D}}}$.  By
  the \ih\ $||\cp||^{\Pi'_{ \cp}}_\y =||\cp||^{\Pi_{ \cp}}_\y \leq
  \ccount{\Delta(\y)}$.  Moreover,  by Lemma~\ref{l:patterns}
 $\Gam'(\x) = \mult{\prodt{\A}{\D}}^{\pair{\p}{\q}}_{\x}$.
Then  the number of free occurrences
  of $\x$ in $\bp$ typed in $\Pi_\bp$ is
 $||\bp||^{\Pi_\bp}_\x=\ccount{\mult{\prodt{\A}{\D}}^{\pair{\p}{\q}}_{\x}} \leq
    \ccount{\mult{\prodt{\A}{\D}}}\leq \ccount{\Del}\leq \ccount{\Del} + \ccount{\sig}$. We conclude since $|\ap|^\Pi_\x  = ||\bp||^{\Pi_\bp}_\x$.

 \item  All other cases follow easily by induction. \qedhere
  \end{itemize}
\end{proof}

This property has an important corollary.
\begin{corollary}\label{cor:finite}
Given a pair $(\Gam, \sig)$, the number of approximate normal forms $\ap$ such that $\Pi\dem\Gam \der \ap:\sigma$ and $\ap=\Ap(\Pi)$ is finite.
\end{corollary}
\begin{proof}
  Let $\Pi \dem \Gam \der \ap :\sig$.  
    By Property~\ref{prop:times},
  the number of typed occurrences of every variable 
 in $\Pi$
    is bounded by $\ccount{\Gam}+\ccount{\sig}$
(we suppose each variable to be either bound or free, but not both,  in $\Ap(\Pi)$, by $\alpha$-conversion). So the total
    number of typed occurrences of variables in $\Pi$ is bounded by an
    integer, let say ${\cal B}$.  By definition of $\Ap(\Pi)$,
    $\Omega$ is the only untyped subterm of $\ap$, then ${\cal B}$ is
    an upper bound for the number of {\it all}  occurrences of
    variables of $\ap$, which turn out to be all typed occurrences of variables of $\ap$ in $\Pi$. 
   
    It is easy to see that ${\cal B}$ is also a bound for the total number of
    axioms of each derivation $\Pi\dem\Gam \der \ap:\sigma$, so the number of
    such derivations is finite and the conclusion follows.
\end{proof}

Now we are able to complete the termination proof.

\thtermination*
\begin{proof}
Lemma~\ref{Lem:main} ($\Rightarrow$) ensures that the outputs of the inhabitation algorithm, called  on $(\Gam, \sig)$, are all of the form $\ap=\Ap(\Pi)$ for some $\Pi \dem \Gam \der \ap:\sigma$. 
By Corollary~\ref{cor:finite} there exist finitely many such  $\ap$'s, and by
 Lemma~\ref{lem:run-termination} producing any of these takes a finite number of steps. Altogether, the inhabitation algorithms always terminates.
\end{proof}


\noindent{\bf Soundness and Completeness.} In order to show Lemma~\ref{Lem:main} we first introduce the following
key notion.  A derivation $\Pi$ is a \deft{left-subtree} of a
derivation $\Sigma$ if either $\Pi=\Sigma$, or
$\Pi \dem \Del \der \uu:\sig$ is the major premise of
some derivation 
$\Sigma' \dem\Del'\der \uu\vv: \tau$, such that $\Sigma'$ is a left-subtree
of $\Sigma$.  \

\begin{property}\label{prop:sound-comple-pattern}
  $(\Gam,  \p)  \in  \Pa{}{\dom{\V}}(\A)$,
 if and only if there is a derivation $\Gam \pder \p: \A$, such that
 $\fv{\p}\cap \V = \emptyset$.
 \end{property}
\begin{proof}
Easy, by checking the rules.
\end{proof}

\noindent {\bf Lemma~\ref{Lem:main}.} 
$\ap\in\K(\Gamma,\sigma)$ $\Leftrightarrow$
$\exists \Pi \dem \Gam \der \ap:\sigma$ such that $\ap=\Ap(\Pi)$.

\begin{proof} 
\begin{description}
\item[$(\Rightarrow)$]
We prove by mutual induction the following statements: 
\begin{enumerate}[label=\alph*)]
\item
$\ap\in\K(\Gamma,\sigma)$ $\Rightarrow$
$\exists \Pi \dem \Gam \der \ap:\sigma$ such that $\ap=\Ap(\Pi)$.
\item 
$\ap\in\KI(\Gam,\A)$ $ \Rightarrow$ $\exists \Pi \dem \Gam \der \ap: \A$ such that $\ap=\Ap(\Pi)$.
\item
$\ap\in\LK{\bp}{\Del}(\Gam, \sig) \donne \tau$  $\Rightarrow$ 
if $\exists \Sigma \dem \Del \der \bp :\sigma$ such that $\bp=\Ap(\Sigma)$,
then  $\exists \Pi \dem \Gam+\Del \der \ap:\tau$ such that $\ap=\Ap(\Pi)$.

\end{enumerate}
Each statement is proved by induction on the 
rules in  Figure~\ref{fig:inhabitation-algorithm}.

\begin{enumerate}[label=\alph*)]
\item 
\begin{itemize}
\item Let the last rule be $(\Abs)$, with conclusion $\lambda\p.\ap\in \K(\Gam,\A\arrow \tau)$ and premises
$\ap\in \K(\Gam + \Del,\tau)$ and
  $(\Delta, \p) \in \Pa{}{\dom{\Gam}}(\A)$.
  By Property~\ref{prop:sound-comple-pattern}, there is a derivation
  $\Del \pder \p: \A$, and we conclude by the \ih\ (a) on
  $\ap\in \K(\Gam + \Del,\tau)$  and the typing
  rule $(\introarrow)$.
\item   Let the last rule be $(\Head)$, with conclusion $\ap\in \K(
\Gamma'+\x :\mult{\tau},\sig)$ 
 and premise
  $\ap\in\LK{\x}{\x:\mult{\tau}}(\Gam', \tau) \donne \sig$,
  where
  $\Gam = \Gamma'+\x :\mult{\tau}$.  Then consider the derivation
  $\Sigma \dem \x:\mult{\tau} \der \x :\tau$ where $\x
  = \Ap(\Sigma)$.  The \ih\ (c) provides $\Pi\dem \Gam'
  + \x:\mult{\tau}\der \ap:\sig$ such that $\ap=\Ap(\Pi)$.
\item Let the last rule be $(\Prod)$, with conclusion $\pair{\ap}{\bp}\in \K(\Gam_0 + \Gam_1,
    \prodt{\A}{\D})$ and premises $\ap \in \KI(\Gam_0, \A)$ and
     $\bp \in \KI( \Gam_1, \D)$, where $\Gam = \Gam_0
     + \Gam_1$ and $\sig = \prodt{\A}{\D}$. Then we conclude by
     the \ih\ (b) and the typing rule ($\trpair$).

\item 
Let the last rule be $(\Subs)$, with conclusion $\bp[\pair{\p}{\q}/\cp] \in \K( \Gam_0+\Gam_1 + \x:\mult{\tau},
    \sig)$ and premises
    $\cp \in \LK{\x}{\x:\mult{\tau}}(\Gam_0, \tau) \donne \fin{\tau},\
    (\Del, \pair{\p}{\q}) \in \Pa{}{\dom{\Gam_0 + \Gam_1 +
    x: \mult{\tau} }}(\mult{\fin{\tau}})$ and $\bp \in \K(\Gam_1
    + \Del, \sig)$, where $\Gam = \Gam_0+\Gam_1
    + \x:\mult{\tau}$.  Since there is
    $\Sigma \dem \x:\mult{\tau}\der\x:\tau$, by \ih\ (c) on
    $\cp \in \LK{\x}{\x:\mult{\tau}}(\Gam_0, \tau) \donne \fin{\tau}$
    there is $\Psi'$ s.t.
    $\Psi' \dem \Gam_{0}+\x:\mult{\tau}\der \cp:\fin{\tau}$, and
    $\cp=\Ap(\Psi')$.  Moreover, by rule $(\many)$ we obtain
    $\Psi\dem\Gam_0+\x:\mult{\tau}\der \cp:\mult{\fin{\tau}}$.  By
    Property~\ref{prop:sound-comple-pattern},
    $(\Del, \pair{\p}{\q}) \in \Pa{}{\dom{\Gam_0 + \Gam_1 +
    x: \mult{\tau} }}(\mult{\fin{\tau}})$ implies there is
    $\Psi''\dem \Delta\pder \pair{\p}{\q}:\mult{\fin{\tau}}$ and
    $\fv{\pair{\p}{\q}} \cap \dom{\Gam_0 + \Gam_1
    + \x: \mult{\tau}}= \es$. Now, by applying the
\ih\ (a) to $\bp \in \K(\Gam_1 + \Del, \sig)$, we get a
derivation $\Pi'
  \dem \Gam_1 + \Del \der \bp:\sig$ such that $\bp = \Ap(\Pi')$.
 We get the required proof $\Pi$ by
  using the typing rule $(\trsub)$ on the premises
  $\Psi$, $\Psi''$ and $\Pi'$.
 \end{itemize}

\item Let the last rule be $(\Unionk)$
with conclusion  $\ap\in\KI(+_{\iI}\Gamma_i,\mult{\sigma_i}_{\iI})$
and premises  $(\ap_i\in\K(\Gam_i,\sigma_i))_{\iI}$ and
  $\uparrow_{\iI}\ap_i$. The proof follows
    from the \ih\ (a) and then the typing rule ($\many$) 
  or the (new) typing rule ($\Omega$). 

\item
\begin{itemize}
\item  Let the last rule be  $(\Final)$,
with conclusion  $\ap\in\LK{\ap}{\Del}(\es, \sigma) \donne \tau$
and premise
$\sigma=\tau$. Suppose
$\Sigma \dem \Del \der \ap: \sigma$. The fact
that the there exists a derivation $\Del + \es \der \ap:\sigma$ is
then straightforward. 

\item Let the last rule be $(\Prefix)$, with conclusion
$\ap\in\LK{\cp}{\Del}(\Gam_0+\Gam_1, \A\arrow \sigma') \donne \tau$ and
premises $\bp \in \KI(\Gam_0,\A)$ and
$\ap\in\LK{\cp\bp}{\Del+\Gam_0}(\Gam_1, \sigma') \donne \tau$,
where $\sig = \A\arrow \sigma'$.  Suppose that there exists
a derivation $\Sigma\dem\Delta\der\cp:\A\arrow\sigma'$ such that
$\cp=\Ap(\Sigma)$.  The \ih\ (b) applied to $\bp \in \KI(\Gam_0,\A)$
provides a derivation $\Psi \dem \Gam_0 \der \bp:\A$, where $\bp
= \Ap(\Psi)$.  The typing rule $(\app)$ with premises $\Sigma$ and
$\Psi$ gives a derivation $\Pi' \dem \Delta + \Gam_0\der\cp\bp:\sig'$,
such that $\cp\bp=\Ap(\Pi')$.
Then, the \ih\ (d) applied to $\ap\in\LK{\cp\bp}{\Del+\Gam_0}(\Gam_1, \sigma') \donne \tau$ provides a derivation
  $\Pi \dem \Delta+\Gam_0+\Gam_1\der \ap:\tau$ such that
$\ap=\Ap(\Pi)$, as required. 

\end{itemize}

\end{enumerate}

\item[$(\Leftarrow)$] We prove by mutual induction the following statements: 

\begin{enumerate}[label=\alph*)]
\item \label{uuno}
Given $\Sigma\dem\Del \der \bp: \tau$ and $\Pi\dem\Gamma \der \ap:\sigma$, 
if $\bp=\Ap(\Sigma)$ and $\ap=\Ap(\Pi)$ are $\cal L$-$\anf s$, and
$\Sigma$ is a left-subtree of $\Pi$,
then there exists $\Gamma'$ s.t. $\Gam =\Gam'+\Delta$ and for every $\Theta$,
$\LK{\ap}{\Del + \Gam'}(\Theta,\sigma) \donne \rho \subseteq 
\LK{\bp}{\Del}(\Theta+\Gam',\tau) \donne \rho$. 
\item \label{ddos}  
$\Pi\dem \Gam  \der \ap:\sigma$ and $\ap=\Ap(\Pi)$ imply $\ap\in \K(\Gam,\sig)$.
\end{enumerate}

Each statement is proved by induction on
the definition of approximate normal forms.
\begin{enumerate}[label=\alph*)]

\item
\begin{itemize}
\item If  $\ap = \x$, then  $\Pi$ is an axiom
 $(\ax)$; $\Sigma$ being a left subtree of
$\Pi$, we get $\Sigma=\Pi$, $\bp=\x$, $\Gamma'=\es$, $\sigma=\tau$ and the inclusion $\LK{\ap}{\Del + \Gam'}(\Theta,\sigma) \donne \rho \subseteq 
\LK{\bp}{\Del}(\Theta+\Gam',\tau) \donne \rho$ trivially holds.
\item If  $\ap = \cp\ap'$, $\cp$ being  an   $\cal L$-$\anf$, then the
last rule of $\Pi$ is an instance of  $(\app)$, 
with premises  $\Pi_1\dem\Gam_1 \der \cp: \A \to \sigma$ and $\Pi_2\dem
\Gam_2 \der \ap': \A$,  so that
 $\Gam'=\Gam_1+\Gam_2$. Moreover, 
 $\Sigma \dem\Del \der \bp: \tau$ is also a left-subtree of $\Pi_1$ and
 $\Pi_2$ comes from $(\Pi^i_2\dem\Gam^i_2 \der \ap': \sig_i)_{\iI}$, where
 $+_{\iI}\Gam^i_2=\Gamma_2$ and $\mult{\sig_i}_{\iI}=\A$.
We have in this case  $\ap'=\bu_{\iI } \Ap(\Pi^i_2)$, where 
by the \ih\ (b), $\Ap(\Pi^i_2)\in \K(\Gam^i_2, \sig_i)$.
  Then  $\LK{\cp\ap'}{\Del + \Gam_1+\Gam_2}(\Theta,\sig) \donne \rho 
\subseteq_{(\Prefix)}\LK{\cp}{\Del + \Gam_1}(\Theta +\Gam_2,\A \to  \sig)\donne \rho  \subseteq_{(a)} 
 \LK{\bp}{\Del}(\Theta+ \Gam_1 + \Gam_2,\tau)\donne \rho $.
 \item Since by hypothesis both $\ap$ and $\bp$ are $\cal L$-$\anf s$, there are no other cases.
\end{itemize}

\item
\begin{itemize}
\item $\ap = \Omega$ does not apply, since $\sig$ is not the empty multiset. 
\item If $\ap$ is an ${\cal L}$-$\anf$, 
then  $\exists \x, \tau$ s.t. $\Gam = \Gam_0 + \x:\mult{\tau}$
and the type derivation $\Sigma \dem \x:\mult{\tau} \der \x: \tau$ is a 
left subtree of $\Pi\dem \Gam_0 + \x:\mult{\tau} \der \ap: \sig$.
Then we have 
$\ap \in \LK{\ap}{\Gam}(\es,\sigma) \donne \sig$
by rule $(\Final)$,
$\LK{\ap}{\Gam_0+\x:\mult{\tau}}(\es,\sigma) \donne \sig
\subseteq
\LK{\x}{\x:\mult{\tau}}(\Gam_0,\tau) \donne \sig$
by Point (a) and
$\LK{\x}{\x:\mult{\tau}}(\Gam_0,\tau) \donne \sig
\subseteq
\K(\Gam_0 + \x:\mult{\tau},\sig)$
by rule $(\Head)$. We thus conclude
$\ap \in \K(\Gam_0 + \x:\mult{\tau},\sig)$. 

\item Otherwise, we analyze all the other cases of ${\cal N}$-$\anf s$. 

\begin{itemize}
\item 
If $\ap = \lambda \p. \bp$ (resp. $\ap = \pair{\bp}{\cp}$) then it is
  easy to conclude by induction, using rule $\Abs$ (resp. $\Prod$).

\item     If $\ap = \cp[\pair{\p}{\q}/\bp]$, then $\cp$ (resp. $\bp$) is an
${\cal N}$
(resp. ${\cal L}$)-approximate normal form.  By construction, $\Pi$ is of the
following form:
\[
        \infer{\Pi' \dem \Gam \der \cp: \sig \sep 
                \Gam|_{\pair{\p}{\q}} \pder \pair{\p}{\q}: \mult{\prodt{\A_1}{\A_2}} \sep
                \infer{\Psi' \dem \Del \der \bp:\prodt{\A_1}{\A_2}}
                {\Psi \dem \Del\der \bp:\mult{\prodt{\A_1}{\A_2}}} 
                }
              {\Gam  \setminus \Gam|_{\pair{\p}{\q}} +\Del \der \cp[\pair{\p}{\q}/\bp]: \sig}
        \]
        By definition $\Ap(\Pi) =\Ap(\Pi')[\pair{\p}{\q}/\Ap(\Psi')]$.
        By the \ih\ (c) $\Ap(\Pi') \in \K(\Gam, \sig)$. Moreover,
        $\bp = \y \cp_1 \ldots \cp_h\ (h \geq 0)$, since it is an
        ${\cal L}$-canonical form, so that
        $\Del = \Del' + (\y:\mult{\tau})$ where
        $\fin{\tau} = \prodt{\A_1}{\A_2}$ and
        $\Sigma\dem\y: \mult{\tau} \der \y: \tau$ is a left
        prefix of $\Psi'\dem\Del \der \bp:\prodt{\A_1}{\A_2}$.
        Therefore, 
        $\Ap(\Psi)= \Ap(\Psi')$ belongs to the set
        $\LK{\Ap(\Psi')}{\Del}(\es,
        \fin{\tau}) \donne \fin{\tau}$, which is
        included in $\LK{\y}{\y:\mult{\tau}}(\Del', \tau) \donne
        \fin{\tau}$ by Point~(a).  
         By Property~\ref{prop:sound-comple-pattern} we also have
        $(\Gam|_{\pair{\p}{\q}} , \pair{\p}{\q}) \in \Pa{}{\dom{\Gam \setminus
            \Gam|_{\pair{\p}{\q}}+\Del}}(\mult{\fin{\tau}})$.  We thus
        obtain
        $\Ap(\Pi)= \Ap(\Pi')[\pair{\p}{\q}/\Ap(\Psi)] \in \K( \Gam
        \setminus \Gam|_{\pair{\p}{\q}} + \Del' +
        (\y:\mult{\tau}), \sig ) = \K( \Gam \setminus
        \Gam|_{\pair{\p}{\q}}+\Del, \sig )$
       by rule $(\Subs)$. \qedhere
\end{itemize}      \end{itemize}
      \end{enumerate}      
    \end{description}
  \end{proof}

\section{Characterizing Solvability}
\label{s:charact}

We are now able to state the main result of this paper, \ie\ the
characterization of the solvability property for the pattern
calculus  $\Lp$.

The logical characterization of canonical forms given in
Section~\ref{s:type-system} through the type assignment system $\Pu$
is a first step in this direction.  In fact, the system $\Pu$ is
complete with respect to solvability, but it
is not sound, as shown in the next theorem.
\begin{theorem}
\label{th:completess}
The set of solvable terms is a proper subset of the set of terms having canonical forms.
\end{theorem}
\begin{proof} \mbox{}
\begin{itemize}
\item (Solvability implies canonicity) If $\s$ is solvable, then
  there is a head context $\hcontext$ such that $\appctx{\hcontext}{\s}$ is
  closed and reduces to $\pair{\uu}{\vv}$, for some $\uu$ and $\vv$.
  Since all pairs are typable, the term $\appctx{\hcontext}{\s}$ is typable by
  Lemma~\ref{lem:red:exp}(\ref{lem:subexp}), so that $\s$ is typable
  too by Lemma~\ref{l:head-contexts}. We conclude
    that $\s$ has canonical form by Theorem~\ref{l:characterization-canonical}.
\item (Canonicity does not imply solvability) Let 
  $\s_1= \lambda \x . \id [\pair{\y}{\z}/\x][\pair{\y'}{\z'}/\x \id]$.
  the term $\s_1$ 
  is canonical, hence typable by
  Theorem~\ref{l:characterization-canonical}. However $\s_1$ is not
  solvable. In fact, it is easy to see that there is no term $\uu$
  such that both $\uu$ and $\uu \id$ reduce to pairs.  Indeed, let
  $\uu \reds \pair{\vv_{1}}{\vv_{2}}$; then
  $\uu\id \reds \pair{\vv_{1}}{\vv_{2}}\id $, which 
  will reduce to $\fail$. \qedhere
\end{itemize}
\end{proof}

However, as explained in the introduction, we can use inhabitation of
system $\Pu$ to completely characterize solvability.  The following
lemma guarantees that the types reflect correctly the structure of the
data.

\begin{lemma}
\label{l:closed-pair}
Let $\s$ be a closed and typable term.
\begin{itemize}
\item If $\s$ has functional type, then $\s$ reduces to an abstraction.
\item 
If $\s$ has  product type, then it  reduces to a pair.
\end{itemize}
\end{lemma}
\begin{proof}
Let $\s$ be a closed and typable term. By Theorem~\ref{l:characterization-canonical} 
we know that 
$\s$ reduces to a (closed) canonical form in ${\cal J}$. The proof is  by induction on the
maximal length of such reduction sequences.
If $\s$ is already a canonical form, we analyze all the cases.
\begin{itemize}
\item  If $\s$ is a variable, then this gives a contradiction with $\s$ closed. 
\item  If $\s$ is an abstraction, then the property trivially holds.
\item If $\s$ is a pair, then the property trivially holds.
\item If $\s$ is an application, then $\s$ necessarily 
  has a head (free)    variable  which
    belongs to the set of  free variables of $\s$, which leads to a
contradiction with $\s$ closed.
\item If $\s=\uu[ \pair{\p_1}{\p_2}/\vv]$ is closed,
  then in particular
  $\vv$ is closed, which leads to a contradiction with
  $\s \in {\cal J}$ implying $\vv
  \in {\cal K}$. So this case is not possible. 
\end{itemize}
Otherwise, there is a reduction sequence $\s \Rew{} \s' \Rewn{} \uu$, where $\uu$ is in ${\cal J}$. 
The term $\s'$ is also closed and typable
by Lemma~\ref{lem:red:exp}(\ref{lem:reduction}), then the \ih\ gives the desired result for
$\s'$, so the property holds also for $\s$.
\end{proof}

The notion of inhabitation can easily be extended to typing
  environments, by defining $\Gam$ inhabited if $\x:\C \in \Gam$
  implies $\C$ is inhabited.  The following lemma shows in
  particular that if the
type of a pattern is inhabited, then
its typing
environment is also inhabited.

\begin{lemma}
  \label{l:inhabited-patterns} \mbox{}
  \begin{enumerate}
 \item \label{l:inhabited-patterns-uno} If $\Pi \dem \Gam \pder \p: \A$ and $\A$ is
    inhabited, then $\Gam$ is also
    inhabited.
    \item \label{l:inhabited-patterns-dos}
      If $\Gam\der \s:\A$ and $\Gam$ is inhabited, 
      then $\A$ is inhabited.
  \end{enumerate}
\end{lemma}

\begin{proof} \mbox{}
  \begin{enumerate}
    \item The proof is by induction on $\p$.
  
      \begin{itemize}
        \item If $\p = \x$ then $\Gam$ is $\x:\A$ with $\A \neq \emul$ or
  it is $\es$. In both
  cases the property is trivial.

  \item 
  If $\p = \pair{\p_1}{\p_2}$, then $\Pi_1 \dem \Gam_1 \pder \p_1:
  \A_1$ and $\Pi_2 \dem \Gam_2 \pder \p_2: \A_2$, where  $\A = \mult{\prodt{\A_1}{\A_2}}$  and $\Gam = \Gam_1 + \Gam_2$. Let us see that $\A_i$ ($i=1,2$)
is inhabited.
Since  $\mult{\prodt{\A_1}{\A_2}}$ is inhabited, so is $\prodt{\A_1}{\A_2}$.  
By Lemma~\ref{l:closed-pair}, the closed term $\s$ inhabiting 
  $\prodt{\A_1}{\A_2}$ reduces to a pair $\pair {\s_1}{\s_2}$.
  We know by Lemma  \ref{lem:red:exp}(\ref{lem:reduction})  that $\der\pair {\s_1}{\s_2} : \prodt{\A_1}{\A_2}$, and we conclude that $\der \s_i: A_i$ ($i=1,2$).
  Now,  by applying  the \ih\ to $\Pi_1$ (resp. $\Pi_2$) we have that for every
  $\x:\A' \in \Gam_1$ (resp $\Gam_2$), $A'$ is inhabited. By linearity of
  $\p$, if $\x:\A' \in \Gam$ then either  $\x:\A'\in\Gam_1$ or
  $\x:\A'\in\Gam_2$ (otherwise stated: $\Gam_1+\Gam_2=\Gam_1;\Gam_2$). Hence we conclude
that for every
$\x:\A' \in \Gam$, $\A'$ is inhabited.
\end{itemize}
\item  For all $\x:\D\in\Gam$, let $\uu_\x$ be a closed term inhabiting $\D$.
  By Lemma \ref{l:substitution-lemma}(\ref{l:substitution}) the closed term obtained by replacing in $\s$
  all occurrences of $\x\in \dom{\Gam}$ by $\uu_\x$  inhabits
  $\A$.  \qedhere
  \end{enumerate}
\end{proof}

In order to simplify the following proofs, let us
introduce a new notation: let $\vec{\A}$ denote a sequence of multiset
types $\A_{1}, ...,\A_{n}$, so that
$\A_{1}\arrow...\arrow\A_{n}\arrow \sig$ will be
abbreviated by $\vec{\A} \arrow \sig$. Note that every type has this
structure, for some multisets $\A_{1}, \ldots, \A_{n}\ (n \geq 0)$ and
type $\sig$. Moreover we will say that $\vec{\A}$ is inhabited if all
its components are inhabited.

\begin{lemma}
  \label{l:extracting-typing}
  Let $\appctx{\hcontext}{\s}$ be such that 
$\Pi \dem \Gam \der \appctx{\hcontext}{\s}: \vec{\A} \arrow \pi$, where  $\Gam$ and $\vec{\A}$  are 
  inhabited. Then there are $\Pi', \Gam',\vec{\C}$ such that $\Pi' \dem \Gam' \der \s: \vec{\C}\arrow \pi$ where $\Gam'$ and $\vec \C$   are  inhabited.
\end{lemma}

\begin{proof} By induction on $\hcontext$.
  \begin{itemize}
  \item If $\hcontext = \Box$, then the property trivially holds.
  \item If $\hcontext = \hcontext'\ \uu$, then $\Gam=\Gam'+\Delta$ and $\Pi$ is:
\[
     \infer{ \seq{\Gam'}{ \appctx{\hcontext'}{\s} : \D\rew \vec{\A}  \rew \pi} \sep\sep 
        \seq{\Del}{\uu : \D}}
      {\seq{\Gam' + \Del}{  \appctx{\hcontext'}{\s}\uu: \vec{\A}  \rew \pi    }}\ (\app)
\]

$\Gam'$ and $\Del$ are inhabited, being  sub-environments of $\Gam$.
By Lemma~\ref{l:inhabited-patterns}(\ref{l:inhabited-patterns-dos})  $\D$ is
inhabited. Then 
the proof follows by the \ih\  on the major premise.
\item  If $\hcontext =\lambda\p.\hcontext'$ then $\Gam=\Gam'\sm\Gam'|_\p$, $\vec \A= \A_0,\vec{\A'}$ and $\Pi$ is:
\[
\infer{ \seq{\Gam'}{ \appctx{\hcontext'}{\s}:{\vec \A'} \rew  \pi} \sep \Gam'|_{\p} \pder  \p:\A_0}  
{\seq{\Gam'\sm\Gam'|_\p}{\lambda \p. \appctx{\hcontext'}{\s}: \A_0\rew{\vec \A'} \rew  \pi}}\ (\introarrow)
\]

Since $\A_0$ is inhabited,
Lemma~\ref{l:inhabited-patterns}(\ref{l:inhabited-patterns-uno})
ensures that $\Gam'|_{\p}$ (and thus $\Gam'$) is
inhabited, too.  The proof follows by the \ih\ on the major
  premise.
\item If $\hcontext = \hcontext'[\p/\uu]$ then $\Gam=\Gam'\sm\Gam'|_\p+\Del$ and $\Pi$ is:
 \[ 
  \infer{\seq{\Gam'}{\appctx{\hcontext'}{\s}:\vec{\A}  \rew \pi} \sep
        \Gam'|_{\p} \pder \p:\D \sep
        \seq{\Del}{\uu:\D}}
      {\seq{(\Gam'\sm \Gam'|_\p) +  \Del}{\appctx{\hcontext'}{\s}[\p/\uu]:\vec{\A}  \rew \pi}}\ (\trsub)
      \]
$\Del$ is  inhabited, being a sub-environment of $\Gam$.
By Lemma~\ref{l:inhabited-patterns}(\ref{l:inhabited-patterns-dos})  $\D$ is
inhabited.
Hence by Lemma \ref{l:inhabited-patterns}(\ref{l:inhabited-patterns-uno})
    $\Gam'|_\p$ (and thus $\Gam'$) is inhabited. 
     Then 
the proof follows by the \ih\  on the major premise. \qedhere
   \end{itemize}
\end{proof}

\begin{theorem}[Characterizing  Solvability]
\label{t:main-result}
A term $\s$ is solvable iff $\Pi \dem \Gam \der \s: \vec{\C} \arrow
\sig$, where  $\sig$ is a product type and $\Gam$ and $\vec{\C}$ are inhabited.
\end{theorem}

\begin{proof} \mbox{}
\begin{itemize}
\item (only if) If $\s$ is solvable, then there
  exists a head context $\hcontext$ such that $\appctx{\hcontext}{\s}$
  is closed and  $\appctx{\hcontext}{\s} \Rewn{} \pair \uu \vv$.
By subject expansion $\der\appctx{\hcontext}{\s}:  \oprod$.
Then Lemma  \ref{l:extracting-typing} allows to conclude.
\item (if) Let $\Gam = \x_1:\A_1, \ldots, \x_k: \A_k\ (k \geq 0)$ and
  $\vec{\C} = \C_1, \ldots, \C_m\ (m \geq 0)$. By hypothesis there exist
  closed terms $\uu_1, \ldots \uu_k, \vv_1, \ldots \vv_m$ such that
  $\es \der \uu_i : \A_i$ and
    $\es \der \vv_j: \C_j$ ($1\leq i \leq k, 1 \leq j\leq m$).  Let
  $\hcontext = (\lambda \x_k \ldots ((\lambda \x_1. \Box) \uu_1) \ldots \uu_k)
  \vv_1 \ldots \vv_m$ be a head context. We have
  $\appctx{\hcontext}{\s}$ closed and $ \es \der \appctx{\hcontext}{\s} : \sig$, where
  $\sig$ is a product type.  This in turn implies that $\appctx{\hcontext}{\s}$
  reduces to a pair, by Lemma~\ref{l:closed-pair}.  Then the term $\s$ 
  is solvable by definition. \qedhere
\end{itemize}
\end{proof}

Our notion of solvability is conservative with respect to that of the $\lambda$-calculus.

\begin{theorem}[Conservativity]\label{thm:con}
A $\lambda$-term $\s$ is solvable in the  $\lambda$-calculus if and only if $\s$ 
is solvable  in the  $\Lp$-calculus.
\end{theorem}

\begin{proof} \mbox{}
\begin{itemize}
\item (if) Let $\s$ be a $\lambda$-term which is not solvable, \ie\ which
  does not have head normal-form. Then $\s$ (seen as a term of our
  calculus) has no \canonical, and thus $\s$ is not typable by
  Theorem~\ref{l:characterization-canonical}.  It turns out that $\s$
  is not solvable in $\Lp$ by Theorem~\ref{t:main-result}.
\item (only if)  Let $\s$ be a solvable $\lambda$-term so that  there  exist  a
  head context $\hcontext$ such that $\appctx{\hcontext}{\s}$ is closed and reduces to $\id$, 
  then it is  easy to construct a head  context $\hcontext'$ such that
  $\appctx{\hcontext'}{\s}$   reduces  to  a   pair  (just   take  $\hcontext'=
  \hcontext\ \pair{\s_1}{\s_2}$ for some terms $\s_1, \s_2$). \qedhere
\end{itemize}
\end{proof}

\section{Conclusion and Further Work}
\label{s:conclusion}

We extend the classical notion of solvability, originally stated for
the $\lambda$-calculus, to a pair pattern calculus.  We provide a logical
characterization of solvable terms by means of typability
{\it and} inhabitation.

An interesting question concerns the consequences
  of changing non-idempotent types to idempotent ones in our
  typing system $\Pu$. Characterization
  of solvability will still need the two
  ingredients typability and inhabitation, however,
  inhabitation will become undecidable, in contrast to our
  decidable inhabitation problem for the non-idempotent
  system $\Pu$. This is consistent with the fact that
  the inhabitation problem for the $\lambda$-calculus
  is undecidable for idempotent types~\cite{Urzyczyn99}, but decidable for
 the non-idempotent ones~\cite{bkdlr14}. 

Notice however that changing the meta-level substitution operator
  to explicit substitutions would not change neither
  the notion nor the characterization of solvability: all the
  explicit substitutions are fully computed in normal forms.
  
Further work will be developed in different directions.  As we already
discussed in Section~\ref{s:calculus}, different definitions of
solvability would be possible, as for example
in~\cite{Garcia-PerezN16}. We explored the one based on a lazy
semantics, but it would be also interesting to obtain a full
characterization based on a strict semantics.  

On the semantical side, it is well known that non-idempotent
intersection types can be used to supply a logical description of the
relational semantics of $\lambda$-calculus~\cite{deC09,paolini12draft}. We
would like to start from our type assignment system for building a
denotational model of the pattern calculus.  Last but not least, a
challenging question is related to the characterization of solvability
in a more general framework of pattern $\lambda$-calculi allowing the
patterns to be dynamic~\cite{JK09}.

\section*{Acknowledgment}
We are grateful to  Sandra Alves and Daniel Ventura
for fruitful discussions. This work was
partially done within the framework of ANR COCA HOLA (16-CE40-004-01).

\renewcommand{\em}{\it}
\bibliographystyle{abbrv}
\bibliography{structure}

\end{document}